\documentclass[reqno, 11pt]{amsart}
\usepackage{mathrsfs,graphicx,url, amssymb}
\usepackage[usenames,dvipsnames]{xcolor}

\usepackage[letterpaper,margin=1in]{geometry}

\usepackage{graphicx}
\usepackage[colorlinks=true, allcolors=blue]{hyperref}
\usepackage{graphicx, tikz} 
\usetikzlibrary{calc, intersections, positioning}
\usepackage{etoolbox}

\usepackage[english]{babel}
\usepackage{enumerate}   

\usepackage[utf8]{inputenc}
\usepackage[T1]{fontenc}
\usepackage{varwidth}
\usepackage[numbers]{natbib}

\numberwithin{equation}{section}

\newtheorem{theorem}{Theorem}[section]
\newtheorem{lemma}[theorem]{Lemma}
\newtheorem{hyp}[theorem]{Hypothesis}
\newtheorem{case}{Case}
\theoremstyle{definition}
\newtheorem{definition}{Definition}[section]

\DeclareMathOperator \re {Re}
\DeclareMathOperator \im {Im}

\newcommand{\supp} {\operatorname{supp}}

\newcommand{\loc}{\operatorname{loc}}
\newcommand{\Real}{\mathbb{R}}
\newcommand{\Complex}{\mathbb{C}}

\newcommand{\mcv}{\mathcal{V}}
\newcommand{\Natural}{\mathbb{N}}
\newcommand{\mcR}{\mathcal{R}}
\newcommand{\ch}{\operatorname{ch}}
\newcommand{\mcV}{\mathcal{V}}

\newcommand{\sing}{\operatorname{sing}}
\newcommand{\sect}{{\mathcal S}}
\newcommand{\ci}{\mathcal{C}}
\newcommand{\rhop}{\rho_+}
\newcommand{\rhom}{\rho_-}
\newcommand{\rhopm}{\rho_{\pm}}

\title[Asymptotic distribution of resonances]{Singularities and asymptotic distribution of resonances for Schr\"{o}dinger operators
in one dimension}
\author{T.J. Christiansen and T. Cunningham}
\address{Department of Mathematics, University of Missouri, Columbia, MO 65211 USA}
\email{christiansent@missouri.edu}
\email{travisdcunningham@gmail.com}
\begin{document}

\begin{abstract}
We obtain new results about the high-energy distribution of resonances for the one-dimensional Schrödinger operator. Our primary result is an upper bound on the density of resonances above any logarithmic curve in terms of the singular support of the potential. We also prove results about the distribution of resonances in sectors away from the real axis, and construct a class of potentials producing multiple sequences of resonances along distinct logarithmic curves, explicitly calculating the asymptotic location of these resonances. The results are unified by the
use of an integral representation of the reflection coefficients, refining methods used in \cite{Froese97} and \cite{Si}.
\end{abstract}
\maketitle

\section{Introduction}
\subsection{Main results}
This paper studies relationships between the singularities of a potential $V\in L^\infty_c(\Real)$ and the high-energy distribution of resonances
of the Schr\"{o}dinger operator $-\frac{d^2}{dx^2}+V$.   It is well known that if $V\in C_c^\infty (\Real)$  then $-\frac{d^2}{dx^2}+V$ has at most finitely many resonances above any logarithmic curve $\{ \im \lambda = -M \log(1+|\lambda|)\}$ (see \cite{Vainberg75,LaxPhillips} or \cite[Section 4.6]{Dyatlov19}).  On the other hand, 
in certain cases it is known that if $V$ has singularities, then $-\frac{d^2}{dx^2}+V$ has a sequence of resonances lying approximately on a 
logarithmic curve; see \cite[Theorem 6]{Zworski87}, or \cite{DatchevN22, Datchev22,Brady23} for some semiclassical results with delta potentials.  This paper explores such connections further, both with some specific classes of examples and with some general theorems.
Fundamental to many of our results is an integral representation for the entries of the scattering matrix, see  (\ref{eq:Tpm}) and (\ref{eq:smat}), and 
an analysis of these integrals.
  Here we build on some techniques and results of 
\cite{Froese97} and \cite{Si}.


For $V\in L^\infty_c(\Real;\Complex)$, denote by $\mcR_V$ the set of resonances of $-\frac{d^2}{dx^2}+V$, repeated according to multiplicity.
A now classic result in the study of resonances is 
\begin{equation}\label{eq:MZas}
\# \{ \lambda_j\in \mcR_V: \; |\lambda_j|<r\} = \frac{2}{\pi}|\ch \supp V|r(1+o(1)) \qquad \text{as} \quad r \rightarrow \infty
\end{equation}
where $\supp V$ is the support of $V$, $\ch\supp V$ is the convex hull of the support, and for an interval $I\subset \Real$, $|I|$ is the 
length of $I$.  See \cite[Theorem 2]{Zworski87}, \cite[Theorem 1.1]{Froese97}, or \cite[Theorem 2.16]{Dyatlov19} for this 
result, and \cite{Regge58} for a special case.  If instead of counting resonances in disks we count those above logarithmic curves, we obtain a result which reflects the role of $\sing \supp V$, the singular support of $V$, rather than the support of $V$, in these neighborhoods of the real axis.  
\begin{theorem}\label{t:ss}
    Let \(V \in L_c^\infty (\Real; \mathbb{C})\). Then for any $M>0$ we have
    \[\overline{\lim_{r\to\infty}} \frac{\# \{ \lambda_j\in \mcR_V: \im \lambda_j>-M\log(|\lambda_j|+1), \pm \re \lambda_j>0\; \text{and}\; |\lambda_j|<r\}}{r} \leq \frac{1}{\pi}|\ch \sing \supp V|.\]
    Moreover, if $M<|\ch \sing\supp V |^{-1}$, then there exists $R>0$ such that
    \[\mathcal{R}_V \cap \{\lambda: \; \im \lambda >-M \log(|\lambda|+1)\} \cap \{|\lambda| > R\} = \emptyset .\]
\end{theorem}
The upper bound obtained in the first part of this theorem is optimal (when considered to hold for any value of $M$). 
 This can be seen from the results of \cite[Theorem 6]{Zworski87} for a class of potentials
$V$ for which $\ch \sing \supp V =\ch \supp V$.  That the optimality holds for potentials for which $|\ch \sing \supp V|<|\ch \supp V|$ is a consequence of our Theorem \ref{t:as}.

We give a  proof of Theorem \ref{t:ss} using a Jensen-type formula for ellipses rather than circles, and an upper bound on the determinant of the 
scattering matrix in appropriate regions of $\Complex$.  
It  may be possible to prove the second part of Theorem \ref{t:ss} using results of \cite{Galkowski17}.  However, the proof we give here  easily follows from the intermediate steps in the proof of the first part of the theorem.

Before stating further results, we introduce some notation which we shall use throughout the paper.  We continue to use $\mcR_V$ to 
denote the resonances of $-\frac{d^2}{dx^2}+V$, repeated with multiplicity.
For \(\Omega\subset\mathbb{C}\), we denote \(n_V (\Omega;r):=\#\{\lambda \in \mathcal{R}_V \cap \Omega : |\lambda| \leq r\}\),
and set  \(n_V (\mathbb{C};r) =: n_V (r)\).   
We take \(\Omega^{\circ}, \bar{\Omega}\), \(\ch \Omega, \: \mathbb{C}\setminus \Omega \) (or \( \Real \setminus \Omega)\) to denote interior, closure, convex hull, and complement, respectively, of a set \(\Omega\) in \(\Complex\) or \(\Real\) depending on the context.

We define a class of potentials of a type considered in \cite[Theorem 6]{Zworski87} (see also \cite{Regge58} and \cite{Stepin07}) although here we allow the potentials to be 
complex-valued.
\begin{definition}\label{d:Vjk}
    We say that \(V \in \mathcal{V}_{j, k}([a,b])\) for integers \(j, k \geq 0\) and an interval \([a, b]\) provided:
    \begin{center}
    \begin{minipage}{5in}

    \begin{enumerate}[(i)]
        \item $\ch \supp V = [a, b]$ and $V \in C^N ([a, b];\Complex)$ for some  $N > j, k$,
        \item For some constants $C_1, C_2 \in \mathbb{C}\setminus \{0\}$, we have
        \[V(x) \sim 
            \begin{cases}
              C_1 (x-a)^j, & \text{as}\ x \rightarrow a^+ \\
              C_2 (b-x)^k, & \text{as}\ x \rightarrow b^-.
            \end{cases}
        \]
    \end{enumerate}
    \end{minipage}
    \end{center}
\end{definition}

Zworski shows that real-valued potentials in \(\mathcal{V}_{j, k} ([a, b])\) produce a sequence of resonances asymptotic to a logarithmic curve, and he gives the form of this sequence explicitly (see \cite[Theorem 6]{Zworski87}). We prove an extension of this result which, in addition to allowing complex-valued $V$, shows that this sequence persists when we perturb by an arbitrary function in $C_c^\infty(\Real)$, even one supported outside of $[a, b]$.  For $M>0$ we define 
\[\hat{L}_M := \{\operatorname{Im}\lambda > -M\log(1 + |\lambda|)\}\]
and notice that $\hat{L}_M$ contains what is sometimes called a logarithmic neighborhood of the real axis.
\begin{theorem}\label{t:as}
    Let \(V \in \mathcal{V}_{j,k}([a,b])\) and \(W \in C_c^\infty (\Real;\mathbb{C})\). Then for any \(M > (j + k + 4) / 2(b - a)\), we can find $R > 0$ such that within \(\Hat{L}_M \cap \{|\lambda| > R\}\) the resonances of \(-\frac{d^2}{dx^2} + V+ W\) form a sequence 
    \[\lambda_{\pm n} = \pm \frac{n \pi}{b-a} \pm \frac{j+k+4}{2(b-a)} \frac{\pi}{2} + i\frac{\log C}{2(b-a)} - i \frac{j+k+4}{2(b-a)} \log \left(\frac{n \pi}{b-a}\right) + \varepsilon_{\pm n}\]
    where \(C = j!k!C_1C_2 / 2^{j+k+4}\) and \(\varepsilon_{\pm n} \rightarrow 0\) as \(n \rightarrow \infty\).
\end{theorem}

The sequence in this theorem lies approximately on the logarithmic curve 
$$\left\{\lambda :\; \im \lambda = \frac{\log |C|}{2(b-a)} - \frac{j+k+4}{2(b-a)}\log |\re \lambda|\right\}$$ and has linear density \(\frac{2}{\pi}(b-a) = \frac{2}{\pi}|\ch \supp V|\). Thus when \(\ch \supp V \subset \ch \supp (V+W)^\circ \), (\ref{eq:MZas}) shows that \(-\frac{d^2}{dx^2} + V+W\) will have infinitely many resonances in \(\mathbb{C}\setminus \hat{L}_M\) and in fact, if $M$ is sufficiently large,
\begin{equation}
    n_{V+W}(\mathbb{C}\setminus \hat{L}_M ;r) = \frac{2}{\pi}\big(|\ch \supp (V+W)| - |\ch \supp V|\big)r\big(1+o(1)\big) \: \text{as} \: r \rightarrow \infty.
\end{equation}

 We note that \cite{Stepin07}  extends \cite[Theorem 6]{Zworski87} to complex-valued potentials supported
in $[a,b]$ while  allowing more general singular behavior at $a$ and $b$.  We extend the results of \cite{Stepin07} in Theorem \ref{t:L1}.

In Section \ref{s:rsum}  we show that one can choose certain
 $V_1\in \mcV_{jk}([a,b])$, $V_2\in \mcv_{kl}([b,c])$ so that the potential $V=V_1+V_2$ has resonances along two distinct logarithmic curves.
 This is made possible by the nature of the singularities of the potential at $a,\; b,$ and $c$.  Moreover, our method allows us to construct potentials with resonances appearing in a sequence of dense clouds of a specified multiplicity concentrated near points along a logarithmic curve.
For other possibilities for such sums under certain additional restrictions on $V_1$ and $V_2$ see Section \ref{s:rsum}.  
 The 
 results of Section \ref{s:rsum} in some sense parallel those for the three delta function Schr\"{o}dinger operators of 
  \cite{Datchev22}.  This is to be expected, since the underlying dynamics of the propagation of singularities should be similar.  
 See \cite{Datchev22} and \cite{GW} for further discussion and references.

We now turn to resonances in sectors away from the real axis.  For $-\pi \leq \theta<\varphi\leq0$ set
$$\sect(\theta, \varphi) := \{\lambda: \theta < \arg \lambda < \varphi\}.$$ A result of \cite{Zworski87,Froese97} shows that if $-\pi <\theta<\varphi < 0$ then $n_V(\sect(\theta, \varphi);r)=o(r)$ as $r\rightarrow \infty$.  Nonetheless, it is known that such sectors may include infinitely many resonances.  In fact,
\cite[Proposition 7]{Zworski87} constructs an explicit example with infinitely many resonances on the imaginary axis.   

This next result gives new information about high-energy resonances in sectors away from $\mathbb{R}$ by showing that the number is unchanged by any interior perturbation of the potential.
To state the theorem, we define 
$$A_V := \{\text{arg}\lambda \in [-\pi, \pi): \lambda \in \mathcal{R}_V\}$$
 and we say that a set \(E \subset \Real_+\) has \emph{finite logarithmic measure} provided
\[\int_E \frac{dx}{x} < \infty.\]
In Section \ref{sec:preliminaries} we show that for any $V, A_V$ is nowhere dense so the condition \(\theta, \varphi \not\in \Bar{A}_V\) in the following theorem is not very restrictive.
\begin{theorem}\label{t:ara}
    Let \(V \in L_{c}^{\infty}(\Real; \mathbb{C})\), let \(\sect = \sect(\theta, \varphi)\) be a sector with \(-\pi < \theta < \varphi < 0\) and \(\theta, \varphi \not\in \Bar{A}_V\), and let \(W \in L_{c}^{\infty}(\Real; \mathbb{C})\) with $\supp W \subset (\ch \supp V)^\circ$. Then there exists \(R > 0\) such that
    \[n_{V+W}(\sect;r) - n_{V+W}(\sect;R) = n_V(\sect;r) - n_V(\sect;R)\]
    for all \(r \geq R, r \not \in E_V\) where $E_V\subset \Real_+$ has finite logarithmic measure and is independent of $\sect$ and $W$.
\end{theorem}
Theorem \ref{t:closeres} shows even more specifically that the resonances of $V$ in this kind of sector $\sect$ are stable under perturbations of the type described here.  
Thus we see that for the high-energy distribution of resonances away from the real axis, it is really the behavior of $V$ at the boundary of its support that is most important.  This is, of course, consistent with the asymptotics (\ref{eq:MZas}) as well
as with the construction in the proof of \cite[Proposition 7]{Zworski87}.  Note the
contrast in the support properties of the perturbation $W$ as compared to $V$ in the hypotheses of Theorem  \ref{t:as} with those of  Theorem \ref{t:ara} or \ref{t:closeres}.  In 
each of these theorems, one should think of $W$ as a perturbation, though it is not necessarily small in norm.

\subsection{Overview and additional connections to previous work}

As for so many results in the study of resonances, we reduce the problem to the 
study of the zeros of an analytic function.   We will use both $\det S_V(-\lambda)$,
where $S_V$ is the scattering matrix and $\im \lambda <0$, and the Fredholm determinant
$\lambda D_V(\lambda)=\lambda \det (I+VR_0(\lambda)\chi)$, where $R_0(\lambda)$ is the free resolvent (see (\ref{eq:free})) and $\chi \in C_c^\infty(\Real)$ 
satisfies $\chi V=V$.   Both $S_V$ and $D_V$ are introduced in Section \ref{sec:preliminaries}.
   The zeros of either of these functions correspond to the resonances of $-\frac{d^2}{dx^2}+V$.  From \cite{Froese97}
we use both that $\lambda D_V(\lambda)$ is an entire function of class C (see Section \ref{sec:preliminaries}) and a 
representation of 
$\det S_V(\lambda)$, (\ref{eq:Sme}).  This representation of the determinant of the scattering matrix and its relation to $D_V(\lambda)$
 allows us to, for example, bound $D_V(\lambda)$ in logarithmic neighborhoods of the real axis in terms of the singular support of 
 $V$.  Preliminary work for this is done in Section \ref{s:spbds}.  This bound, along with a Jensen-type formula for ellipses, is what is needed to prove Theorem \ref{t:ss} in Section \ref{s:ss}.  
 
 Section \ref{s:extension} proves three extensions of \cite[Theorem 6]{Zworski87} (see also \cite{Stepin07}), including Theorem \ref{t:as}.   Each of these expands the 
 class of potentials for which one can explicitly calculate (up to small error) a sequence of resonances asymptotic to a logarithmic curve.  Section \ref{s:rsum} 
 also produces examples of such potentials, though we find more complicated behavior as well.  For potentials of the 
 type considered in Theorem \ref{t:as}, (\ref{eq:Sme}) allows us to find the leading terms of $\det S_V(-\lambda)$ when $-T \log (|\lambda|+1)<\im \lambda<0$.  This and an application of Hardy's method, see Lemma \ref{l:Hm1}, are the main ingredients in the proof of Theorem \ref{t:as}. 
 
 Our results in Section \ref{s:logc} give a variety of explicit examples of resonances generated by diffraction of singularities in the homogenous (nonsemiclassical) setting. Theorem \ref{t:as} in particular can be used to give more general examples of the optimality of the resonance-free
region for diffractive trapping obtained in \cite[Theorem 1]{GW} (see \cite[Theorem 2]{GW} where a potential in $\mathcal{V}_{k,l}([0, L])$ is used to illustrate that optimality). The paper \cite{Datchev22} (see also \cite{DatchevN22,Brady23}) gives some results analogous to our Section \ref{s:rsum} on the existence of resonances along distinct logarithmic curves, but in the context of $\delta$-function potentials; although we will not explore the connection in this paper, we note that the dynamical picture of propagation of singularities in these two settings is very similar. 
  For much more
on the appearance of sequences of resonances along logarithmic curves and related phenomena in a variety of scattering theoretic contexts, the 
reader may consult \cite[Section 1.2]{GW}, \cite[Introduction]{Datchev22}, and references therein. Along with those in \cite{Datchev22}, our results are among the most precise in that we calculate expressions for the asymptotic location of the resonances, allowing one to see explicitly how the sequences vary with the parameters defining the potentials.

It is a result of Vainberg  and Lax-Phillips \cite{Vainberg75, LaxPhillips}
 (which holds in more general settings, see also \cite[Section 4.6]{Dyatlov19}) that if $V\in C_c^\infty(\Real;\Real)$ then
$-\frac{d^2}{dx^2}+V$ has only finitely many resonances above any logarithmic curve.  Galkowski \cite{Galkowski17} quantifies this, by 
relating the smoothing properties of the wave group to the existence of specific logarithmic neighborhoods of the real axis with only finitely many resonances.

We prove Theorem \ref{t:ara} in Section \ref{s:awayfromR}.  Although all of the results of this paper use an intermediate step of \cite{Froese97} (see also \cite{Si}) as a starting point, it is the proof of Theorem \ref{t:ara} which has the most in common with \cite{Froese97}, using results from the theory of entire functions of exponential type.

 An appendix  provides some applications of a well-known method of G.H. Hardy on 
locating the solutions of certain transcendental equations which we use repeatedly in Section \ref{s:logc}. 

In addition to the papers previously mentioned, \cite{K1,K2} and references therein prove further estimates on resonances in one dimension.

  \vspace{2mm}
\noindent\textbf{Acknowledgments.} We are very grateful to the Prison Math Project and in particular Tian An Wong both for general support and for
 handling the typing of the original manuscript of this paper.   Thanks also to Carlo Beenakker and Dan Cunningham  for logistical help.

\section{Preliminaries}\label{sec:preliminaries}
For $\im\lambda > 0$ we set $R_0 (\lambda) := (-\frac{d^2}{dx^2} - \lambda ^ 2)^{-1}:L^2(\Real)\rightarrow L^2(\Real)$.  The free resolvent
$R_0$ has explicit integral kernel
\begin{equation}\label{eq:free}
    R_0(\lambda ; x,y) = \frac{i}{2\lambda}e^{i\lambda|x-y|}
\end{equation}
and hence continues meromorphically to $\mathbb{C}$ as an operator from $L_c^2(\Real)$ to $L_{\loc}^2(\Real)$. Given $V \in L_c^\infty (\Real;\mathbb{C})$, the resolvent $R_V(\lambda) := (-\frac{d^2}{dx^2} + V - \lambda^2)^{-1}$, initially defined for $\im\lambda \gg 1$, likewise has a meromorphic continuation to $\mathbb{C}$ as an operator from $L_c^2(\Real)$ to $L_{\loc}^2(\Real)$; the poles of this meromorphic continuation are called resonances.

It is well-known (see \cite[Chapter 2]{Dyatlov19} and \cite{Froese97}) that resonances can also be identified as the zeros of the Fredholm determinant,
$$D_V(\lambda) := \det(I+VR_0(\lambda)\chi),$$
where $\chi \in L_c^\infty(\Real)$ satisfies $\chi V = V$, using that $VR_0(\lambda)\chi$ is a meromorphic family of trace class operators. (We note that much of our notation follows \cite{Froese97} and \cite[Chapter 2]{Dyatlov19}, but the notation for the resolvents is that of \cite{Dyatlov19}.)
The function $D_V$ has the following properties:
\begin{center}
\begin{varwidth}{0.8\textwidth}
\begin{enumerate}[(i)]
    \item \(\det S_V(\lambda) = D_V(-\lambda)/D_V(\lambda)\) where $S_V(\lambda)$ is the scattering matrix.
    \item The function $f_V(\lambda) := \lambda D_V(\lambda)$ is an entire function of exponential type with indicator function 
    \[h_{f_V}(\theta) = 
        \begin{cases}
          0 & \ \theta \in [0, \pi] \\
          2|\ch \supp V||\sin \theta| & \ \theta \in [-\pi, 0].
        \end{cases}
    \]
    \item There are constants $C_0,\; C_1,\; C_2>0$ so that in  $\im \lambda \geq 0, |\lambda| > C_0$ we have $C_1 \leq |D_V(\lambda)| \leq C_2$ 
    (see \cite[Theorem 2.17]{Dyatlov19}); in particular
    \[\int_{-\infty}^{\infty} \frac{\log^+ |f_V(t)|dt}{1+t^2} < \infty.\]
\end{enumerate}
\end{varwidth}
\end{center}

In \cite[Chapter 16]{Levin96}, entire functions satisfying (ii) and the last part of (iii) are said to be of Class C. Applying \cite[Section 16.1, Theorem 2]{Levin96} to $f_V$ and using (i) above, we obtain the following application of these facts:

\begin{lemma}\label{l:Sasymp}
    Given $V \in L_c^\infty (\Real;\mathbb{C})$, we can find a union of disks $\ci_V = \bigcup\limits_{i=1}^\infty B(z_i;r_i)$ satisfying $\sum\limits_{i=1}^\infty \frac{r_i}{|z_i|} < \infty$, and such that
    \begin{equation}\label{eq:Sasympt}
        \log |\det S_V(-\lambda)| = 2|\ch \supp V||\im \lambda| + o(|\lambda|)
    \end{equation}
    for $\im \lambda \leq 0$ and $\lambda \not \in \ci_V$.
\end{lemma}

Although the exceptional disks in the set $\ci_V$ are not effectively constructed, because resonances are the zeros of $\det S_V(-\lambda)$ in $\im \lambda \leq 0$, it is obvious from (\ref{eq:Sasympt}) that $\mathcal{R}_V \subset \ci_V \cup B(0, C)$ for some large $C$. A particular consequence of this is that the set $A_V = \{\arg\lambda \in [-\pi, \pi): \lambda \in \mathcal{R}_V\}$ is nowhere dense. Indeed, if $\Bar{A}_V$ contained a non-empty interval $I$, we choose $R > 0$ so that $\sum\limits_{|z_i| > R} \frac{r_i}{|z_i|} < |I|/4$. Since $\frac{r_i}{|z_i|} = \sin \frac{\alpha_i}{2}$ where $\alpha_i$ is the magnitude of the opening of the smallest sector containing the disk $B(z_i; r_i)$, and since $\frac{1}{2} < \frac{\sin \theta}{\theta}$ for small $\theta$, we see that $\sum\limits_{|z_i| > R} \alpha_i < |I|$. This is clearly a contradiction.

Next we define
\begin{align}\label{eq:fpm}
    f_\pm (x, \lambda) & := e^{\pm i\lambda x} R_V(-\lambda)e^{\mp i \lambda x} V,\\
\label{eq:Tpm}
    \rhopm (-\lambda) &:= \int e^{\mp 2i\lambda x} V(x) (1-f_\pm (x, \lambda))dx.
\end{align}
and
\[
\mathcal{T}_\pm(-\lambda) := \int V(x)(1-f_\pm(x, \lambda))dx.
\]

We will sometimes use superscripts $f_\pm^V$, etc., to denote dependence upon the potential. In \cite[Lemma 3.3]{Froese97}, Froese shows that
\begin{equation}\label{eq:fbd}
    |f_\pm (\bullet, \lambda)| \leq \frac{C}{|\lambda|}
\end{equation}
for $\im\lambda \leq 0$ and $|\lambda| > C_1$, and his proof shows that the choice of $C, C_1$ depends only on $\|V\|_{L^{1}}$.

The point of the definitions above is that the following representation of the scattering matrix holds:
\begin{equation}\label{eq:smat}
    S_V(\lambda) = I + \frac{1}{2i\lambda}
    \begin{pmatrix}
    \mathcal{T}_+(\lambda) & \rhom(\lambda)\\
    \rhop(\lambda) & \mathcal{T}_-(\lambda)
    \end{pmatrix}
    .
\end{equation}
Thus the functions $\rhopm$ which we shall study are $2i\lambda$ times the reflection coefficients.

We therefore obtain the following useful expansion for the determinant of the scattering matrix:
\begin{equation}\label{eq:Sme}
    \det S_V(-\lambda) = 1 +\frac{1}{4\lambda^2}\rhom(-\lambda)\rhop(-\lambda)+O(|\lambda|^{-1})
\end{equation}
for $\im\lambda \leq 0$ and $|\lambda| > C_0$ for some $C_0$. See \cite{Froese97} and \cite[Chapter 2]{Dyatlov19} for more on this representation.

We note that although it is standard to work with potentials $V \in L_c^\infty(\Real)$ (and in fact we need this for some of the bounds we make), the definition of $R_V(\lambda)$, and hence of $f_\pm$ and $\rhopm$ are equally valid when $V \in L_c^1(\Real)$ and we will work briefly with such a $V$ in Theorem \ref{t:L1}.

The definition of $f_\pm $ from (\ref{eq:fpm}) and a resolvent identity yield
\begin{equation}\label{eq:fpmB}
f_{\pm}=B_{\pm} (V )-B_{\pm}(Vf_\pm)
\end{equation}
where
\begin{equation}\label{eq:Bpm}
B_\pm = B_\pm (\lambda) := e^{\pm i \lambda x} R_0(-\lambda) e^{\mp i \lambda \bullet}.
\end{equation}

Another tool we will need is the following important property of $\rhopm$ involving sums of potentials; this observation will be crucial for us in most of the remaining sections.
\begin{lemma}\label{l:Tsum}
    Let $V, W \in L_c^\infty (\Real;\mathbb{C})$ where $\ch \supp V = [a, b]$, $\ch \supp W = [a_0, b_0]$, with $a \leq a_0 < b_0 \leq b$. Then there exists $C_0 > 0$ such that
    \begin{align*}\rhom^{V+W} (-\lambda)&  = \rhom^V(-\lambda) + e^{2i\lambda b_0}O(1)\\
    \rhop^{V+W} (-\lambda) &= \rhop^V(-\lambda) + e^{-2i\lambda a_0}O(1)
    \end{align*}
    when $\im \lambda \leq 0, |\lambda| > C_0$.
\end{lemma}
\begin{proof}
     From (\ref{eq:Tpm}) we may write
    \begin{multline*}
    \rhom^{V+W}(-\lambda) - \rhom^V(-\lambda) \\= \int_{a_0}^{b_0} e^{2i\lambda x}W(x)(1-f_-^{V+W}(x, \lambda))dx 
    + \int_a^b e^{2i\lambda x} V(x)(f_-^V (x, \lambda)-f_-^{V+W} (x, \lambda))dx.
    \end{multline*}
    Using (\ref{eq:fbd}), the first term on the right is clearly $e^{2i\lambda b_0}O(1)$. For the second term,
    (\ref{eq:fpmB}) gives
    \begin{equation}\label{eq:fdiff}
        f_-^{V+W} - f_-^V = B_-(W) - B_-(Wf_-^{V+W}) + B_-(V(f_-^V - f_-^{V+W})).
    \end{equation}
    Defining
    \begin{align*}
    F(x, \lambda) & := V(x)(f_-^{V+W}(x, \lambda) - f_-^V(x, \lambda))\\
    G(x, \lambda) & := V(x)B_-(W)(x, \lambda) - V(x)B_-(Wf_-^{V+W})(x, \lambda),
    \end{align*}
    we note the easy bound
    \begin{equation}\label{eq:Gbd}
        G(x, \lambda) = e^{2i\lambda(b_0-x)}O(1) \quad \text{in } \im\lambda \leq 0, |\lambda| > C_0,
    \end{equation}
    and write $V$ times (\ref{eq:fdiff}) as
    \begin{equation}\label{eq:cat}
        F(x, \lambda) = G(x, \lambda) - V(x)B_-(F)(x, \lambda).
    \end{equation}
    
    We need to show that 
    \begin{equation}
        \int_a^b e^{2i\lambda x}F(x, \lambda)dx = e^{2i\lambda b_0}O(1)
    \end{equation}\label{eq:Artie}
    in $\im\lambda \leq 0, |\lambda| > C_0$. If $\chi \in L_c^\infty$ and $\chi V = V$, then (\ref{eq:cat}) says that
    \[(I+VB_-\chi)F = G.\]
    Using the explicit integral kernel for the free resolvent (\ref{eq:free}), it is easy to see that in $\im\lambda \leq 0, |\lambda| > C_0$, we have
    \[\| VB_\pm\chi \| _{L^2 \rightarrow L^2} \leq \frac{C}{|\lambda|},\]
    so we can invert $(I+VB_\pm \chi)$ and $\|(I+VB_\pm \chi)^{-1}\|_{L^2 \rightarrow L^2} \leq C$ in $\im\lambda \leq 0,$ when $ |\lambda| > C_0$. We now calculate
    \begin{align*}
    \int_a^b e^{2i\lambda x} F(x, \lambda)dx & = \int_a^b e^{2i\lambda x}\left[(I+VB_-\chi)^{-1}G\right](x, \lambda)dx \\
    & = \int_a^b \left[(I+VB_+\chi)^{-1}(e^{2i\lambda\bullet}G)\right](x, \lambda)dx,
    \end{align*}
    where we used that $e^{2i\lambda x}(I+VB_-\chi )^{-1}e^{-2i\lambda\bullet} = (I+VB_+\chi )^{-1}$. Therefore using (\ref{eq:Gbd}),
    \[\left|e^{-2i\lambda b_0} \int_a^b e^{2i\lambda x} F(x, \lambda)dx\right| = \left|\int_a^b \left[(I+VB_+\chi)^{-1}(e^{2i\lambda (\bullet - b_0)}G)\right](x, \lambda)dx\right| \leq C\]
    in $\im\lambda \leq 0, |\lambda| > C_0$. This proves (\ref{eq:Artie}) and hence the lemma for $\rhom$. The proof for $\rhop$ is analogous.
\end{proof}

Lemma \ref{l:Tsum} is used in Section \ref{s:awayfromR} through the following corollary giving a bound on the difference of two scattering determinants.

\begin{lemma}\label{l:Sdiff}
    Let $V$ and $W$ be as in Lemma \ref{l:Tsum}. If $M := \max\{b_0 - a, b-a_0\}$, then there exist $C,\;C_0 > 0$ such that 
    \[|\det S_{V+W}(-\lambda) - \det S_V(-\lambda)| \leq C e^{2M|\im\lambda |}\]
    for $\im \lambda \leq 0, |\lambda| > C_0$.
\end{lemma}
\begin{proof}
     Using (\ref{eq:Sme}) we have
    \begin{multline*}\det S_{V+W} (-\lambda) - \det S_V(-\lambda) = \frac{1}{8\lambda^2}\left(\rhom^{V+W} (-\lambda) - \rhom^V(-\lambda)\right)\left(\rhop^{V+W}(-\lambda) + \rhop^V(-\lambda)\right)\\
    + \frac{1}{8\lambda^2}\left(\rhop^{V+W}(-\lambda) - \rhop^V(-\lambda)\right)\left(\rhom^{V+W}(-\lambda) + \rhom^V(-\lambda)\right)+O(|\lambda|^{-1}).
    \end{multline*}
    Since (\ref{eq:fbd}) gives us the bounds
    \begin{align*}
    \rhom^{V+W}(-\lambda) + \rhom^V(-\lambda) &= e^{2i\lambda b}O(1)\\
    \rhop^{V+W}(-\lambda) + \rhop^V(-\lambda) &= e^{-2i\lambda a}O(1),
    \end{align*}
    one easily concludes using Lemma \ref{l:Tsum} and taking $|\lambda|$ large enough.
\end{proof}

\section{Resonances along logarithmic curves}\label{s:logc}

In this section we outline a  method which shows how functions in the class $\mathcal{V}_{j, k} ([a, b])$ from  Definition \ref{d:Vjk} can be used to construct potentials with sequences of resonances along logarithmic curves.

\subsection{Preliminary results for $\rhopm$}    \label{s:spbds}

We begin here by developing the tools which allow us to relate the asymptotics of the determinant of the scattering matrix to certain properties of the potential. The potentials we consider in this section consist of sums of functions in $L_c^\infty(\Real),\; C_c^\infty(\Real)$, and $\mathcal{V}_{j,k}([a,b])$ and in particular will have varying degrees of regularity on intervals bounded by certain distinguished points of the support. The following definition is made to capture that concept, and gives us a shorthand notation for the proofs of our technical lemmas.

\begin{definition}
    If $a_0 \leq a \leq b \leq c$ and $0 \leq j \leq k \leq l \leq \infty$, we say that a function $U: \Real \rightarrow \mathbb{C}$ is in $\mathcal{U}(a_0, a, b, c; j, k, l)$ provided:
    \begin{center}
    \begin{varwidth}{0.8\textwidth}
    \begin{enumerate}[(i)]
        \item $U \in C^\infty ([a_0, a]) \cap C^{N_1}([a,b]) \cap C^{N_2}([b,c])$ for some $N_1 > k, N_2 > l,$
        \item $U$ is smooth across $a_0$ and continuous up to (but not necessarily including) order $j$ at $a$, order $k$ at $b$, and order $l$ at $c$.
    \end{enumerate}
    \end{varwidth}
    \end{center}
Moreover, if $a_0 = a, a = b$, or $b = c$, we require $U$ to satisfy only the more lenient of the restrictions at that point.
\end{definition}

We note that we have not required these functions to have compact support, as we will also use the definition to describe the regularity properties of the functions $f_\pm$, see Lemma \ref{l:abc}.

We shall see that for some such $U$, Lemma \ref{l:UTs}
combined with (\ref{eq:Sme}) shows that the leading behavior of $\det S_U(-\lambda)$ in $\im \lambda \leq 0$ is determined by the Fourier transform of the potential.  A first step in this direction, which will allow us to integrate by parts in the definition of $\rhom$, is the following lemma.
There is an analogous result for $f_+$.


\begin{lemma}\label{l:abc}
    Let $U \in \mathcal{U}(a_0, a, b, c; j, k, l)$ with $\ch \supp U=[a_0, c]$. Then there exists $C_0 > 0$ such that in $\im \lambda \leq 0, |\lambda| > C_0$, we have $f_-(\bullet, \lambda) = f_-^U(\bullet, \lambda) \in \mathcal{U}(a_0, a, b, c; j+2, k+2, l+2)$ and the following expansions hold:
    \begin{center}
    \begin{varwidth}{0.85\textwidth}
    \begin{enumerate}[(i)]
        \item In $[a_0, a], \partial_x^d f_-(\bullet, \lambda) = O(|\lambda|^{-1})$ for $d \geq 0$.
        \\
        \item In $[a, b], \partial_x^d f_-(\bullet, \lambda) = O(|\lambda|^{-1})$ for $0 \leq d \leq j+1$,
        \\
        $\partial_x^d f_-(\bullet, \lambda) = e^{-2i \lambda (x-a)} O(|\lambda|^{d-j-2}) + O(|\lambda|^{-1})$ for $j+2 \leq d \leq k+1.$
        \\
        \item In $[b,c], \partial_x^d f_-(\bullet, \lambda) = O(|\lambda|^{-1})$ for $0 \leq d \leq j+1$,
        \\
        $\partial_x^d f_-(\bullet, \lambda) = e^{-2i \lambda (x-a)} O(|\lambda|^{d-j-2}) + O(|\lambda|^{-1})$ for $j+2 \leq d \leq k+1$,
        \\
        $\partial_x^d f_-(\bullet, \lambda) = e^{-2i \lambda (x-b)} O(|\lambda|^{d-k-2}) + e^{-2i \lambda (x-a)} O(|\lambda|^{d-j-2}) + O(|\lambda|^{-1})$ for $k+2 \leq d \leq l+1.$
    \end{enumerate}
    \end{varwidth}
    \end{center}
\end{lemma}
\begin{proof}
    That $f_- \in \mathcal{U}(a_0, a, b, c; j+2, k+2, l+2)$ follows from its definition in terms of the resolvent and standard elliptic regularity theory. 
    Using (\ref{eq:fpmB}), (\ref{eq:Bpm}) and  the explicit integral kernel for the free resolvent (\ref{eq:free}) we easily obtain that for $x \in [a_0, c]$,
    \begin{equation}\label{eq:start}
        \partial_x f_-(x, \lambda) = - \int_{a_0}^x e^{-2i\lambda (x-y)} U(y)dy + \int_{a_0}^x e^{-2i\lambda (x-y)}(Uf_-)(y)dy
    \end{equation}
    in $\im \lambda \leq 0, |\lambda| > C_0$. Using (\ref{eq:start}) and integration by parts, the proof of (i), as well as the first parts of (ii) and (iii), is an elementary induction argument beginning from (\ref{eq:fbd}).

    To complete the proof of (ii), when $x \in [a, c]$ we split the integrals in (\ref{eq:start}) at $y=a$, differentiate, and then integrate by parts to get
    \begin{align*}
    \partial_x^{j+1} f_-(x, \lambda)&  = - \int_{a_0}^a e^{-2i\lambda (x-y)} U^{(j)}(y)dy - \int_a^x e^{-2i\lambda (x-y)}U^{(j)}(y)dy \\
     & \;\;\;+ \int_{a_0}^a e^{-2i\lambda (x-y)} (Uf_-)^{(j)}(y)dy + \int_a^x e^{-2i\lambda (x-y)}(Uf_-)^{(j)}(y)dy \\
    & = e^{-2i\lambda (x-a)}C(\lambda) -  \int_a^x e^{-2i\lambda (x-y)} U^{(j)}(y)dy + \int_a^x e^{-2i\lambda (x-y)}(Uf_-)^{(j)}(y)dy
    \end{align*}
in $\im \lambda \leq 0, |\lambda| > C_0$, where $C(\lambda) = O(|\lambda|^{-1})$ and is independent of $x \in [a,c]$. Using this equation, integration by parts, and the bounds made in the first part of the proof, we inductively produce the remainder of (ii), as well as the second part of (iii).

    The last part of (iii) is obtained from (i) and (ii) in analogous fashion; this completes the proof of the lemma.
\end{proof}

\begin{lemma}\label{l:UTs}
    For $U \in \mathcal{U}(a_0, a, b, c; j, k, l)$ with $\ch \supp U = [a_0, c]$, we have
    \begin{align*}
    \rhom(-\lambda) & = \hat{U}(-2\lambda) + e^{2i\lambda a}O(|\lambda|^{-j-2}) + e^{2i\lambda b}O(|\lambda|^{-k-2}) + e^{2i\lambda c} O(|\lambda|^{-l-2})\\
    \rhop(-\lambda) & = \hat{U}(2\lambda) + e^{-2i\lambda a_0}O(|\lambda|^{-\infty} )+ e^{-2i\lambda a}O(|\lambda|^{-j-2})
    \end{align*}
    in $\im \lambda \leq 0, |\lambda| > C_0$.
\end{lemma}
\begin{proof}
The proofs of these expansions are very similar, with our lopsided assumptions on $U$ making the one for $\rhop$ the simpler to obtain. Thus we will give the proof for $\rhom$.

The definition of $\rhom$ in (\ref{eq:Tpm}) gives
\[\rhom^U (-\lambda) = \hat{U}(-2\lambda) - \int_{a_0}^c e^{2i\lambda x} U(x) f_-^U (x, \lambda) dx.\]
Splitting the integral at $x=a$ and $x=b$, Lemma \ref{l:UTs} is obtained by integrating by parts an appropriate number of times in each term, noting cancellations due to regularity of the integrand at $a$ and $b$, and then appealing to Lemma \ref{l:abc} to make the required bounds. This gives the desired expansion for $\rhom$.
\end{proof}

\subsection{Extensions of a theorem of Zworski}\label{s:extension}
Here we give three extensions of \cite[Theorem 6]{Zworski87}, including Theorem \ref{t:as}.   Each of these extends the class of potentials for which one can explicitly calculate a sequence of resonances asymptotic to a logarithmic curve.

We first note that the Fourier transform of $V \in \mathcal{V}_{j,k} ([a,b])$ has a particularly nice form. Indeed, a simple integration by parts calculation shows that
\begin{equation}\label{eq:vhat}
    \hat{V}(\pm 2\lambda) = \frac{k!C_2}{(\mp 2i\lambda)^{k+1}}e^{\mp 2i\lambda b}(1+o(1)) + \frac{j!C_1}{(\pm 2i\lambda)^{j+1}}e^{\mp2i\lambda a}(1+o(1))
\end{equation}
in $\im \lambda \leq 0$, where the decaying terms are had using \cite[Lemma 2.2]{Titchmarsh26} which shows that when $f \in C([0,1])$,
\[\int_0^1 e^{-i\lambda x} f(x)dx = o(1)\]
uniformly as $|\lambda| \rightarrow \infty$ through $\im\lambda \leq 0$.

\begin{proof}[Proof of Theorem \ref{t:as}]
    Without loss of generality we assume $j \leq k$. If $\ch \supp W = [a_0, b_0]$, we may also take $a_0 \leq a < b \leq b_0$, for if we prove this case the general one follows by writing $V+W = (V-P) + (W+P)$ for some $P \in C_c^\infty$ with $\ch \supp P = [a,b]$ and noting that $(V-P) \in \mathcal{V}_{j,k}([a,b])$ with the same endpoint behavior as $V$.

    With these agreements, we see that $V+W \in \mathcal{U}(a_0, a, b, b_0; j, k, \infty)$ and therefore using Lemma \ref{l:UTs} and (\ref{eq:vhat}) we obtain
    \begin{align*}\rhom(-\lambda) & = \frac{k!C_2}{(2i\lambda)^{k+1}} e^{2i\lambda b}(1+o(1)) + e^{2i\lambda a}O(|\lambda|^{-j-1}) + e^{2i\lambda b_0} O(|\lambda|^{-\infty}) \\
    \rhop(-\lambda) &= \frac{j!C_1}{(2i\lambda)^{j+1}} e^{-2i\lambda a}(1+o(1)) + e^{-2i\lambda b}O(|\lambda|^{-k-1}) + e^{-2i\lambda a_0} O(|\lambda|^{-\infty})
    \end{align*}
    for $\im\lambda \leq 0, |\lambda| > C_0$. The expansion (\ref{eq:Sme}) then shows, after absorbing some terms into one another, that
    \begin{equation}\label{eq:smsum}
        \det S_{V+W}(-\lambda) = 1-\frac{j!k!C_1C_2}{(2i\lambda)^{j+k+4}} e^{2(b-a)i\lambda}(1+o(1)) 
        + e^{2(b_0-a_0)i\lambda}O(|\lambda|^{-\infty})+O(|\lambda|^{-1})
    \end{equation}
    for $\im \lambda \leq 0, |\lambda| > C_0$. When $M > (j+k+4)/2(b-a)$, as $|\lambda| \rightarrow \infty$ through $\hat{L}_M \cap \{\im \lambda \leq 0 \}$ all but the first two terms on the right in (\ref{eq:smsum}) decay so
    \[ \det S_{V+W}(-\lambda) = 1-\frac{j!k!C_1C_2}{(2i\lambda)^{j+k+4}} e^{2(b-a)i\lambda}(1+o(1)) +o(1)\]
    in this region. Just as in the proof of \cite[Lemma 5]{Zworski87}, we conclude by appealing to Hardy's method (see Lemma \ref{l:Hm1} in the appendix), which produces from this expansion the sequence of resonances in the conclusion of the theorem.
\end{proof}

Theorem \ref{t:as} shows that the asymptotic behavior of the sequence of resonances produced by $V \in \mathcal{V}_{j,k}([a,b])$ is unaffected by an arbitrary smooth perturbation. This next extension shows that the asymptotic behavior of the  sequence is likewise unaffected by perturbations with less regularity which are
supported on an 
appropriate subset of the support of $V$.

\begin{theorem}\label{t:lri}
    Let $V \in \mathcal{V}_{j,k}([a,b])$ and $W \in L_c^\infty(\Real;\mathbb{C})$ with $\supp W \subset (a,b)$. Then there exists $M_0 > 0$ so that for any $M > M_0$ there exists $R > 0$ with
    \[\mathcal{R}_{V+W} \cap (\mathbb{C}  \setminus \hat{L}_M) \cap \{|\lambda| > R\} = \emptyset .\]
    Moreover, if
    \[\supp W \subset \left(a + \frac{(j+1)(b-a)}{j+k+4}, b-\frac{(k+1)(b-a)}{j+k+4}\right)\]
    then there exists $R > 0$ such that within $\{|\lambda| > R\}$ the resonances of $-\frac{d^2}{dx^2} +V+W$ form a sequence
    \[\lambda_{\pm n} = \pm \frac{n \pi}{b-a} \pm \frac{j+k+4}{2(b-a)}\frac{\pi}{2} + i\frac{\log C}{2(b-a)} - i\frac{j+k+4}{2(b-a)}\log\left(\frac{n\pi}{b-a}\right) + \varepsilon_{\pm n}\]
    where $C = j!k!C_1C_2/2^{j+k+4}$ and $\varepsilon_{\pm n} \rightarrow 0$ as $n \rightarrow \infty$.
\end{theorem}

\begin{proof}
    Let $\ch \supp W=[a_1, b_1]$ and without loss of generality assume $j \leq k$. We use Lemma \ref{l:Tsum} to see that in $\im \lambda \leq 0, |\lambda| > C_0$,
    \begin{align*}\rhom^{V+W}(-\lambda) &= \rhom^V(-\lambda) + e^{2i\lambda b_1}O(1)\\
    \rhop^{V+W}(-\lambda)& = \rhop^V(-\lambda) + e^{-2i\lambda a_1}O(1).
    \end{align*}
    Since $V \in \mathcal{U}(a, a, b, b;j, k, k)$ we apply Lemma \ref{l:UTs} and then (\ref{eq:vhat}) to get
    \begin{align*}
    \rhom^{V+W}(-\lambda) &= \frac{k!C_2}{(2i\lambda)^{k+1}} e^{2i\lambda b}(1+o(1)) + e^{2i\lambda a}O(|\lambda|^{-j-1}) + e^{2i\lambda b_1} O(1),\\
    \rhop^{V+W}(-\lambda) &= \frac{j!C_1}{(2i\lambda)^{j+1}} e^{-2i\lambda a}(1+o(1)) + e^{-2i\lambda b}O(|\lambda|^{-k-1}) + e^{-2i\lambda a_1} O(1),
    \end{align*}
    and (\ref{eq:Sme}) gives
    \begin{equation}\label{eq:dog}
    \begin{split}
        \det S_{V+W}(-\lambda) = 1-\frac{j!k!C_1C_2}{(2i\lambda)^{j+k+4}} e^{2(b-a)i\lambda}(1+o(1)) + e^{2(b_1-a_1)i\lambda}O(|\lambda|^{-2}) \\
        + e^{2(b-a_1)i\lambda}O(|\lambda|^{-k-3}) + e^{2(b_1-a)i\lambda}O(|\lambda|^{-j-3}) + O(|\lambda|^{-1})
    \end{split}
    \end{equation}
    in $\im \lambda \leq 0, |\lambda|>C_0$. Define
    \[M_0 := \max \left\{ \frac{j+k+2}{2[(b-a)-(b_1-a_1)]}, \frac{j+1}{2(a_1-a)}, \frac{k+1}{2(b-b_1)}, \frac{j+k+4}{2(b-a)}\right\}.\]
    Using  (\ref{eq:dog}) we see that as $|\lambda| \rightarrow \infty$ through $\mathbb{C} \setminus \hat{L}_M = \{\im \lambda \leq -M\log(1+|\lambda|)\}$,
    if $M>M_0$ we have
    \[\left|\frac{(2i\lambda)^{j+k+4}}{j!k!C_1C_2} e^{-2(b-a)i\lambda} \det S_{V+W}(-\lambda) + 1 \right| = o(1)\]
    and therefore $\det S_{V+W}(-\bullet)$ cannot have a zero there once $|\lambda|$ is large enough.   This proves the first part of the theorem.

    Next we define
    \[M_1 := \min \left\{ \frac{1}{b_1-a_1}, \frac{j+3}{2(b_1-a)}, \frac{k+3}{2(b-a_1)}\right\}\]
    and note that when the support of $W$ is restricted as in the statement of the theorem, we have $M_1 > (j+k+4)/2(b-a) = M_0$. Thus, if $M_1>T>M_0$, in $\hat{L}_T\cap \{\im \lambda \leq 0\}$ all but the first two terms on the right in (\ref{eq:dog}) decay as $|\lambda| \rightarrow \infty$ and just as in the proof of Theorem \ref{t:as} we find that the large resonances there form the claimed sequence. In $\mathbb{C}\setminus \hat{L}_T$ on the other hand, there are only finitely many resonances for the same reason as in part one of the proof, so the proof is complete.
\end{proof}

We note that the first part of Theorem \ref{t:lri} improves modestly upon Theorem \ref{t:ara} for this class of potentials. Examples of $W$ which do not satisfy the support condition and for which the conclusion of the theorem is then false may be constructed using the methods of the next subsection. Indeed, taking $W \in \mathcal{V}_{j',k'}([a_1,b_1])$ the parameters may be chosen so that
\[(a,b) \supset [a_1,b_1] \supset \left(a + \frac{(j+1)(b-a)}{j+k+4}, b-\frac{(k+1)(b-a)}{j+k+4}\right)\]
and the potential $V+W$ then produces \emph{two} sequences of resonances.

We now show how our method can be used to extend the main result of \cite{Stepin07}. There the authors directly analyze certain solutions to $(-\frac{d^2}{dx^2} + V - \lambda^2)u = 0$ which allows them to consider potentials of a form
\begin{equation}\label{eq:stepinp}
    \begin{aligned}
    V(x) = (x-a)^\mu V_0(x)(b-x)^\nu,  \; \supp V_0 \subset [a,b], \quad V_0(a)V_0(b) \not= 0,
    \end{aligned}
\end{equation}
where $V_0 \in C^m([a, b]), m \geq \max\{\mu, \nu\} + 2,$ and $\mu, \nu \geq 0$ are not necessarily integers, giving yet another extension of \cite[Theorem 6]{Zworski87}. Their potentials are required to have an absolutely integrable derivative (see \cite[Section 2]{Stepin07}), and therefore the cases $0>\mu,\nu>-1$ cannot be included in the theorem. In addition to extending the main result in \cite{Stepin07}, the following gives an application of our method to potentials with infinite singularities at the endpoints of 
the support.

\begin{theorem}\label{t:L1}
    Let $V$ be of the form (\ref{eq:stepinp}) with $0 > \mu, \nu >-1$ and $V_0 \in C^1 ([a,b])$. Then there exists $R>0$ such that within $\{|\lambda|>R\}$ the resonances of $-\frac{d^2}{dx^2}+V$ form a sequence
    \[\lambda_{\pm n} = \pm\frac{n\pi}{b-a} \pm \frac{\mu+\nu+4}{2(b-a)}\frac{\pi}{2} + i\frac{\log C}{2(b-a)} - i\frac{\mu+\nu+4}{2(b-a)}\log\left(\frac{n\pi}{b-a}\right) + \varepsilon_{\pm n}\]
    where $C=(b-a)^{\mu+\nu}V_0(a)V_0(b)\Gamma(\mu+1)\Gamma(\nu+1)/2^{\mu+\nu+4}$ and $\varepsilon_{\pm n} \rightarrow 0$ as $n \rightarrow \infty$.
\end{theorem}

\begin{proof}
    We first note that $V \in L^1 (a,b)$ and calculating the Fourier transform one obtains that for any $\delta > 0$,
    \begin{align*}
    \hat{V}(-2\lambda) &= \frac{(b-a)^\mu V_0(b)\Gamma(\nu+1)}{(2i\lambda)^{\nu+1}}e^{2i\lambda b}(1+o(1)) + e^{2i\lambda a} O(1)\\ 
    \hat{V}(2\lambda) &= \frac{(b-a)^\nu V_0(a)\Gamma(\mu+1)}{(2i\lambda)^{\mu+1}}e^{-2i\lambda a}(1+o(1)) + e^{-2i\lambda b} O(1)
    \end{align*}
        uniformly as $|\lambda| \rightarrow \infty$ through $\im \lambda \leq -\delta$.

    Using the definitions of $\rhopm$ and (\ref{eq:fbd}) we get from this that in $\im \lambda \leq -\delta$,
    \begin{align*}
    \rhom(-\lambda) &= \frac{(b-a)^\mu V_0(b)\Gamma(\nu+1)}{(2i\lambda)^{\nu+1}}e^{2i\lambda b}(1+o(1)) + e^{2i\lambda a} O(1) \\
    \rhop(-\lambda) &= \frac{(b-a)^\nu V_0(a)\Gamma(\mu+1)}{(2i\lambda)^{\mu+1}}e^{-2i\lambda a}(1+o(1)) + e^{-2i\lambda b} O(1).
    \end{align*}
    Hence (\ref{eq:Sme}) shows that for any $\varepsilon > 0$, we have
    \[\det S_V(-\lambda) = 1-\frac{(b-a)^{\mu+\nu}V_0(a)V_0(b)\Gamma(\mu+1)\Gamma(\nu+1)}{(2i\lambda)^{\mu+\nu+4}} e^{2(b-a)i\lambda}(1+o(1))+O(|\lambda|^{-1})\]
    in $\im \lambda \leq -\varepsilon\log(1+|\lambda|), |\lambda| > C_0$. We conclude using Lemma \ref{l:Hm1}, noting that for sufficiently small $\varepsilon$, there may be at most finitely many resonances in the complement of this region.
\end{proof}

We mention here that we could easily give versions of Theorems \ref{t:as} and \ref{t:lri} for $V$ as in Theorem \ref{t:L1} and it is interesting to notice what happens to the support condition on $W$ in Theorem \ref{t:lri} as $j,k \rightarrow -1$.

\subsection{A class of potentials producing multiple sequences of resonances along logarithmic curves}\label{s:rsum}
This section contains another extension of \cite[Theorem 6]{Zworski87}.  
We will determine the asymptotic distribution of resonances for a potential $V$  which satisfies the following
conditions.
 \begin{hyp}\label{h:V} Let $j,k,l\in \Natural_0$, with 
 $0\leq j\leq k\leq l < \infty, $  $a<b<c$, $V \in C^{N_1}([a,b]) \cap C^{N_2}([b,c])$ for some $N_1>k,\; N_2>l$, $\ch \supp V=[a,c]$, 
  and
\[V(x) \sim 
    \begin{cases}
      C_1 (x-a)^j, & \text{as}\ x \rightarrow a+ \\
      C_2 (b-x)^k, & \text{as}\ x \rightarrow b- \\
      C_3 (x-b)^k, & \text{as}\ x \rightarrow b+ \\
      C_4 (c-x)^l, & \text{as}\ x \rightarrow c-
    \end{cases}
\]
for some constants  $C_i \in \mathbb{C}\setminus \{0\}$ with $C_2 \not= (-1)^k C_3$. 
\end{hyp}
We note that $V=V_1+V_2$ where $V_1 \in \mathcal{V}_{j,k}([a,b]), V_2 \in \mathcal{V}_{k,l}([b,c])$, and the condition on $C_2 $ and $C_3$ ensures
$V$ has a singularity at $b$.
In principle, one could use the techniques we develop here to consider sums of more potentials of these same types, resulting (in some cases)
in more strings of resonances.

Here we will see that there are three cases, determined by the relative size of three quantities determined by $a,\;b,$ and $c$ and by $j,\;k,$ and $l$.  Our results are very similar to the case of the three delta functions potential of \cite{Datchev22}.


Clearly $V \in \mathcal{U}(a,a,b,c;j,k,l)$.  Setting
\begin{align}
A= b-a,\; & \alpha = j+k+4, \; C_A= \frac{j!k!C_1(C_2+(-1)^{k+1}C_3)}{{2}^\alpha}\\
B= c-a,\; & \beta= j+l+4,\; C_B= \frac{j!l!C_1C_4}{2^{\beta}}
\end{align}
and applying the methods of Section \ref{s:extension} yields the following asymptotic representation of the determinant of the scattering matrix:
\begin{equation}\label{eq:S2p}
    \det S_V(-\lambda) = 1-\frac{C_A}{(i\lambda)^\alpha} e^{2Ai\lambda} (1+o(1)) - \frac{C_B}{(i\lambda)^\beta} e^{2Bi\lambda} (1+o(1)) + O(|\lambda|^{-1})
\end{equation}
in $\im \lambda \leq 0, |\lambda| > C_0$.  Here we used that $j\leq k\leq l$ made it possible to absorb one additional term which would occur with more general hypotheses.  We shall see below that the large $|\lambda|$ distribution of resonances depends 
on the relative sizes of $T_1$ and $T_2$, where 
\begin{equation}
T_1 := \frac{\alpha}{2A} = \frac{j+k+4}{2(b-a)}, \; T_2 := \frac{\beta}{2B} = \frac{j+l+4}{2(c-a)},\; \text{and}\;
T_3 := \frac{\beta - \alpha}{2(B-A)} = \frac{l-k}{2(c-b)}.
\end{equation}
For simplicity below we shall study resonances in $\re \lambda\geq 0$; resonances in $\re \lambda<0$ can be
 found by using that $\lambda $ is a resonance of $-\frac{d^2}{dx^2}+V$ if and only if $-\overline{\lambda}$ is a resonance of $-\frac{d^2}{dx^2}+\overline{V}.$ There are three cases.

\begin{case}\label{c:1}($T_1 > T_2$)  Let $V$ satisfy the hypotheses \ref{h:V}, and suppose $T_1>T_2$.  Then there is an $N\in \Natural$, $R>0$, so
that the resonances of $-\frac{d^2}{dx^2}+V$ in $|\lambda| >R$, $\re \lambda \geq 0$ are given by 
 \begin{equation}\label{eq:c1} \lambda_{n} =  \frac{n \pi}{c-a}  +\frac{j+l+4}{2(c-a)}\frac{\pi}{2} + i\frac{\log C_B}{2(c-a)} - i\frac{j+l+4}{2(c-a)}\log\left(\frac{n\pi}{c-a}\right) + \varepsilon_{ n},\; n>N
 \end{equation}
 with $\varepsilon_{ n}\rightarrow 0$ as $ n \rightarrow \infty$.
 \end{case}

    \begin{proof}
    It is easy to check that the assumption $T_1>T_2$ implies $T_2>T_3$. Pick $T$ with $T_2>T>T_3$. In $\hat{L}_T$ we find
    \[|\lambda^{-\alpha} e^{2Ai\lambda}| < |\lambda|^{-\alpha}(1+|\lambda|)^{2AT} = o(1)\]
    as $|\lambda| \rightarrow \infty$; similarly the term involving $\lambda^{-\beta} e^{2Bi\lambda}$ decays as well so we determine from (\ref{eq:S2p}) that there are only finitely many resonances in $\hat{L}_T$.

    In $\mathbb{C}\setminus \hat{L}_T = \{\im \lambda \leq -T\log(1+|\lambda|)\}$ on the other hand, we see that
    \[|\lambda^{\beta-\alpha} e^{2(A-B)i\lambda}| \leq |\lambda|^{\beta-\alpha}(1+|\lambda|)^{-2(B-A)T} = o(1)\]
    so $\lambda^{-\alpha} e^{2Ai\lambda} = \lambda^{-\beta} e^{2Bi\lambda}o(1)$ and therefore
    \[\det S_V(-\lambda) = 1 - \frac{C_B}{(i\lambda)^\beta} e^{2Bi\lambda}(1+o(1)) + O(|\lambda|^{-1})\]
    in $\mathbb{C} \setminus \hat{L}_T$. We conclude using Lemma \ref{l:Hm1} that the resonances of $-\frac{d^2}{dx^2} + V$ in $|\lambda|>R$, $\re \lambda \geq 0$
     form a sequence given by (\ref{eq:c1}).\end{proof}
    Notice that in Case \ref{c:1} we obtain asymptotically the same sequence as is given by Theorem \ref{t:as} if there were no singularity at $b$. The sequence lies approximately on $\{\im \lambda = \frac{\log|C_B|}{2(c-a)} - T_2\log|\re\lambda|\}$ and has linear density $\frac{B}{\pi} = \frac{1}{\pi}|\ch\supp V|$.


\begin{case}[$T_2>T_1$] \label{c:2}
    \:
    Let $V$ satisfy the hypotheses \ref{h:V}, and suppose $T_2>T_1$.  Then there is an $N\in \Natural$, $R>0$, so
that the resonances of $-\frac{d^2}{dx^2}+V$ in $|\lambda| >R$, $\re \lambda \geq 0$ are given by 
 \begin{equation}\label{eq:c2a} \lambda_{n,1} =\frac{n \pi}{b-a} +\frac{j+k+4}{2(b-a)}\frac{\pi}{2} + i\frac{\log C_A}{2(b-a)} - i\frac{j+k+4}{2(b-a)}\log\left(\frac{n\pi}{b-a}\right) + \varepsilon_{ n,1},\;   n>N
 \end{equation}
 and 
 \begin{equation}\label{eq:c2b} \lambda_{n,2} =\frac{n \pi}{c-b} + \frac{l-k}{2(c-b)}\frac{\pi}{2} + i\frac{\log [-C_A^{-1}C_B]}{2(c-b)} - i\frac{l-k}{2(c-b)}\log\left(\frac{n\pi}{c-b}\right) + \varepsilon_{n,2},\; n>N
 \end{equation}
 with $\varepsilon_{ n,j}\rightarrow 0$ as $ n \rightarrow \infty$.
 \end{case}
    
    \begin{proof}
    The assumption $T_2>T_1$ implies  $T_3>T_2$.  Let $T\in \Real$ satisfy $T_2>T>T_1$. In $\hat{L}_T$ the term involving $\lambda^{-\beta}e^{2Bi\lambda}$ in (\ref{eq:S2p})  decays, so
    \[\det S_V(-\lambda) = 1 - \frac{C_A}{(i\lambda)^\alpha} e^{2Ai\lambda}(1+o(1)) + o(1)\]
    as $|\lambda| \rightarrow \infty$ through $\hat{L}_T \cap \{\im \lambda \leq 0\}$. Thus, in this region the resonances
    with $\re \lambda >0$ form the sequence (\ref{eq:c2a}).

    In $\mathbb{C}\setminus \hat{L}_T$ we see that $|\lambda^\alpha e^{-2Ai\lambda}| = o(1)$, so multiplying (\ref{eq:S2p}) by $-C_A^{-1}(i\lambda)^\alpha e^{-2Ai\lambda}$ we find
    \[-C_A^{-1}(i\lambda)^\alpha e^{-2Ai\lambda} \det S_V(-\lambda) = 1 + \frac{C_A^{-1}C_B}{(i\lambda)^{\beta-\alpha}} e^{2(B-A)i\lambda}(1+o(1)) + o(1).\]
    Lemma \ref{l:Hm1} gives us a \emph{second} sequence in $\re \lambda>0$ given by (\ref{eq:c2b}).
\end{proof}
We note the sequence (\ref{eq:c2a}) lies 
     approximately on $\{\im \lambda = \frac{\log|C_A|}{2(b-a)} - T_1\log|\re \lambda|\}$ and has linear density $\frac{A}{\pi} = \frac{1}{\pi}(b-a)$, while 
     the sequence (\ref{eq:c2b}) lies
    approximately on $\{\im \lambda = \frac{\log|C_A^{-1}C_B|}{2(c-b)} - T_3\log|\re\lambda|\}$ and has linear density $\frac{B-A}{\pi} = \frac{1}{\pi}(c-b)$. The densities of the two sequences sum to $\frac{1}{\pi}|\ch \supp V|$, as they must (recalling there are similar sequences in the left half plane).

\begin{case}[$T_1=T_2$] If $V$ satisfies hypotheses \ref{h:V} and $T_1=T_2$, then there is an $N\in \Natural$ so that for 
   $m = 1, \dots, j+l+4$, and all $n \geq N$,  $-\frac{d^2}{dx^2}+V$ has  resonances
        \begin{multline}\label{e:ms}
        \lambda_{n,m} =  \frac{(j+l+4)n \pi}{c-a} + \frac{(j+l+4)}{2(c-a)}\frac{\pi}{2} - i\frac{(j+l+4)}{2(c-a)}\log z_m \\
        - i\frac{(j+l+4)}{2(c-a)}\log\left(\frac{(j+l+4)n\pi}{c-a}\right) + \varepsilon_{n,m} 
        \end{multline}
        where $\{z_m\}_{m=1}^{j+l+4}$ are the roots of the equation $C_B z^{j+l+4} + C_A z^{j+k+4} - 1 = 0$ repeated with multiplicity, and $\varepsilon_{n,m} \rightarrow 0$ as $n \rightarrow \infty$.   These may not be in one-to-one correspondence with all the (sufficiently large) resonances in $\re \lambda>0$.
     \end{case}
    
    \begin{proof}
    When $T_1=T_2$, we find $T_3=T_1$ and no term of (\ref{eq:S2p}) dominates.    In this case we must use Lemma \ref{l:Hm2}, which 
    shows that there are zeros of $\det S_{V}(-\lambda)$ at points described by (\ref{e:ms}).\end{proof}

    In contrast with Case \ref{c:2}, these sequences lie approximately on logarithmic curves which differ at most by a finite shift.
    Notice moreover that the linear density of each sequence $\{\lambda_{ n,m}\}_{n=N}^{\infty}$ is equal to $\frac{c-a}{(j+l+4)\pi}$ and that together, therefore, the set of resonances found  has linear density $\frac{1}{\pi}|\ch \supp V|$ (compare with (\ref{eq:MZas})); we make no claim, however, that we have found all but finitely many of the resonances.

    We also note that the $j+l+4$ sequences could each lie approximately on \emph{distinct} logarithmic curves, or we could have multiple of them lying on the same curve, depending upon the multiplicities of the $\{z_m\}$. Indeed, by tuning the parameters $j,k,l,C_A,C_B$, so that the equation $C_B z^{j+l+4} + C_A z^{j+k+4} - 1 = 0$ has a solution of a desired multiplicity $M>1$, one constructs a potential which has resonances asymptotically in dense clouds, each of this same multiplicity. For instance, suppose the parameters are chosen so that $z_0$ is a solution of $C_B z^{j+l+4} + C_A z^{j+k+4} - 1 = 0$ of multiplicity $M>1$ (it is easy to see that this is possible for any specified multiplicity). Then for $m = 1, \dots, M$, and all $n \geq N$, we will have the resonances
    \begin{align*} \lambda_{n,m} & =  \frac{(j+l+4)n \pi}{c-a} + \frac{(j+l+4)}{2(c-a)}\frac{\pi}{2} - i\frac{(j+l+4)}{2(c-a)}\log z_0
     - i\frac{(j+l+4)}{2(c-a)}\log\left(\frac{(j+l+4)n\pi}{c-a}\right) + \varepsilon_{n,m} \\ &
    =: \gamma_{n} + \varepsilon_{n,m}.
    \end{align*}
As these sequences differ only in the small remainder term $\varepsilon_{n,m}$, for each $n \geq N$ we find a cloud of $M$ resonances concentrated near the points $\gamma_{n}$, and since $\varepsilon_{n,m} \rightarrow 0$, these clouds will be more tightly packed the larger we take $n$.


\section{A bound for resonances in terms of the singular support}\label{s:ss}
In this section we prove Theorem \ref{t:ss}.  
In order to do so, we first recall a complex analysis result for estimating the number of zeros of an analytic function in an ellipse. The following is taken from \cite[Ch. IX, Sec. C, Corollary, pg. 61]{Koosis92}.

\begin{lemma}\label{l:koosis}(\cite{Koosis92})
    Let $a$ and $\gamma > 0$, and suppose that $F(z)$ is analytic inside and on the ellipse
    \[E_\gamma := \{z = a\cosh(\gamma + i\theta): \theta \in [0, 2\pi]\}.\]

    If $0 \leq \eta < \gamma$ and $N_\eta$ denotes the number of zeros of $F$ (counting multiplicities) inside and on the ellipse $E_\eta$, then
    \[(\gamma - \eta)N_\eta \leq \frac{1}{2\pi}\int_0^{2\pi} \log |F(a\cosh(\gamma + i\theta))|d\theta - \frac{1}{\pi}\int_{-a}^a \frac{\log|F(t)|}{\sqrt{a^2-t^2}}dt.\]
\end{lemma}

The result stated in \cite{Koosis92} is actually slightly weaker than this, but the statement above is implied without any change to the proof. We will apply this result to $D_U(\lambda)$ (which is always analytic away from $\lambda = 0$; see Section \ref{sec:preliminaries}) and to ellipses which are shifted and rescaled so as to occupy larger and larger portions of the region above a logarithmic curve.  
To estimate the resulting terms we will need  the following bound. We set $x_- := \max\{0, -x\}$.

\begin{lemma}\label{l:Dub}
    Let $U \in L_c^\infty(\Real;\mathbb{C})$ and $W \in C_c^\infty(\Real;\mathbb{C})$. Then for any $M>0$, there exists $C \in \Real$ and $C_0 > 0$ such that
    \[\log|D_{U+W}(\lambda)| \leq C+2|\ch \supp U|(\im \lambda)_-\]
    in $\hat{L}_M$ when $|\lambda| > C_0$.
\end{lemma}

\begin{proof}
    If $\ch \supp U = [a,b]$, $\ch\supp W = [a_1,b_1]$, we may assume that $a_1 \leq a < b \leq b_1$, for otherwise we write $U+W = (U-P)+(W+P)$ for some $P \in C_c^\infty(\Real)$ with $\ch \supp P = [a,b]$ and apply the proof to $U_1 = U-P$ and $W_1 = W+P$.

    From Section \ref{sec:preliminaries}, we know that for some $C_0, C_2, |D_{U+W}(\lambda)| \leq C_2$ in $\im\lambda \geq 0, |\lambda| > C_0$, and
    \begin{equation}\label{eq:Dsum}
        |D_{U+W}(\lambda)| = |D_{U+W}(-\lambda)| |\det S_{U+W}(-\lambda)| \leq C_2|\det S_{U+W}(-\lambda)|
    \end{equation}
    in $\im\lambda \leq 0, |\lambda| > C_0$.

    Now, Lemma \ref{l:Tsum} (applied with $V_{\ref{l:Tsum}}=W$, $W_{\ref{l:Tsum}}=U$) shows that
    \begin{align*}\rhom^{U+W}(-\lambda) & = \rhom^W(-\lambda) + e^{2i\lambda b}O(1) \\
    \rhop^{U+W}(-\lambda) &= \rhop^W(-\lambda) + e^{-2i\lambda a}O(1).
    \end{align*}

    Using Lemma \ref{l:UTs} for the terms $\rhopm^W$, we obtain that
    \begin{equation}\label{eq:tprod}
        \rhom^{U+W}(-\lambda)\rhop^{U+W}(-\lambda) = e^{2i\lambda|\ch\supp U|}O(1) + e^{2i\lambda|\ch\supp W|}O(|\lambda|^{-\infty}).
    \end{equation}
    in $\im \lambda \leq 0, |\lambda| > C_0$. Therefore using (\ref{eq:Sme}) we see that for any $M>0$, in $\hat{L}_M \cap \{\im \lambda \leq 0\}$ 
    if  $|\lambda| > C_0$, one has
    \[|\det S_{U+W}(-\lambda)| \leq Ce^{-2|\ch \supp U|\im \lambda}.\]
    Pairing this with (\ref{eq:Dsum}) proves the lemma.
\end{proof}

Given $V \in L_c^\infty(\Real)$ and $\varepsilon>0$ we can write
\begin{equation}\label{eq:Uassum}
V=U+W\; \text{with $U \in L_c^\infty(\Real)$, $W \in C_c^\infty(\Real)$, and  $|\ch \supp  U| \leq |\ch \sing \supp V | + 2\varepsilon$}
\end{equation}
where of course the choice of $U$ and $W$ depends on $\varepsilon$.
Using this decomposition of $V$ and (\ref{eq:tprod}) we prove the second part of Theorem \ref{t:ss}.
Indeed given $M<|\ch \sing \supp V|^{-1}$, choose $\varepsilon > 0$
 so that $$M < (|\ch \sing\supp V| + 2\varepsilon)^{-1} \leq |\ch \supp U|^{-1}$$ where we have written $V=U+W$ as 
 in (\ref{eq:Uassum}). Using (\ref{eq:Sme}) and (\ref{eq:tprod}) we see that $\det S_V(-\lambda) = 1+o(1)$ in $\hat{L}_M \cap \{\im \lambda \leq 0\}$ and therefore $\mathcal{R}_V \cap \hat{L}_M \cap \{|\lambda| > R\} = \emptyset$ for sufficiently large $R$.

 Since in the decomposition (\ref{eq:Uassum}) $\varepsilon>0$ is arbitrary the rest of Theorem \ref{t:ss} follows from Theorem \ref{t:ssv2}.   For the proof recall that 
 $\hat{L}_M^\pm := \hat{L}_M \cap \{\pm \re\lambda > 0\}$, and define
 \[L_M := \{\lambda:\im\lambda = -M\log(1+|\lambda|)\}.\]

 \begin{theorem}\label{t:ssv2}
     Let $U \in L_c^\infty(\Real;\mathbb{C})$ and $W \in C_c^\infty(\Real;\mathbb{C})$. Then for any $M>0$ we have
     \[\overline{\lim_{r\rightarrow\infty}}\frac{n_{U+W}(\hat{L}_M^\pm ; r)}{r} \leq \frac{1}{\pi}|\ch \supp U|.\]
 \end{theorem}

 \begin{proof}
     It suffices to prove the claim for $\hat{L}_M^+$ only and we will do so by contradiction. To that end, suppose that for some $M>0$ and some $\varepsilon_0 > 0$ we can find a sequence $\{r_k\}_{k=1}^\infty$ with $r_k \nearrow \infty$ and
     \begin{equation}\label{eq:lba}
         \frac{n_{U+W}(\hat{L}_M^+ ;2r_k)}{2r_k} > \frac{1}{\pi}|\ch\supp U| + \varepsilon_0
     \end{equation}
     for all $k$.

     Fix the positive constants $\varepsilon_1, \varepsilon_2, \varepsilon_3, \gamma_0$, and $\delta_0$ such that $\varepsilon_2 < \varepsilon_3 := (\frac{2}{\pi}|\ch\supp (U+W)| + 1)^{-1} \varepsilon_0$ and
     \begin{equation}\label{eq:d0}
         \left(\frac{2}{\pi}|\ch\supp U| + \varepsilon_0\right)\left(1-\varepsilon_1\right)\frac{\gamma}{\sinh\gamma} - \frac{2}{\pi}\left|\ch\supp U\right| \geq \delta_0
     \end{equation}
     for $\gamma \in [0, \gamma_0].$ We will use (\ref{eq:d0}) later to produce a contradiction, but with these quantities now fixed, we take some $T\gg M$ (we will place a more precise restriction on $T$ below) and set
     \begin{equation}\label{eq:etak}
         \eta_k := \sinh^{-1}\left(\frac{T\log(1+(1+\varepsilon_2)r_k)}{r_k}\right)\; \text{and} \; \gamma_k := \varepsilon_1^{-1}\eta_k.
     \end{equation}
     Note that $0 < \eta_k < \gamma_k \rightarrow 0$ as $k \rightarrow \infty$. The purpose of these definitions is to ensure that the sequence of ellipses
     \[E_k := \{\lambda = (1+\varepsilon_2)r_k + r_k\cosh(\eta_k + i\theta): \theta \in [0, 2\pi]\}\]
     contain larger and larger portions of $\hat{L}_M^+$ (see Figure \ref{fig:enter-label}).
     
\begin{figure}
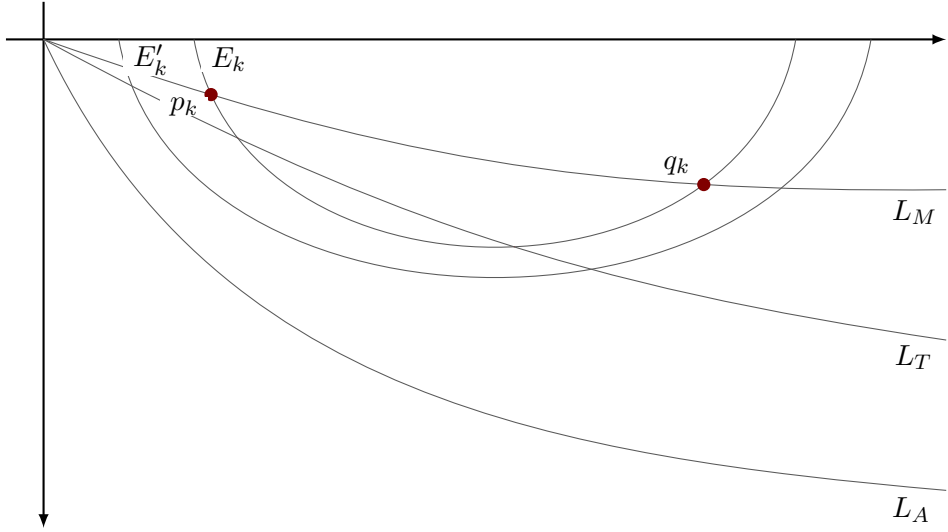

    \centering
\tikz{
\draw[-latex,  thick] (-0.5,0)--(12,0);
\draw[-latex,  thick] (0,0.5)--(0,-6.5);

\draw[name path=A, black!65] (1,0) to[bend right=80, looseness=1.1] (11,0);
\draw[name path=B, black!65] (2,0) to[bend right=80, looseness=1.2] (10,0);

\draw[name path=l1, black!65] (0,0) to[bend right=10, looseness=1] node[pos=.97, below, text=black] {$L_M$} (12,-2);
\draw[name path=l2, black!65] (0,0) to[bend right=10, looseness=1] node[pos=.97, below, text=black] {$L_T$} (12,-4);
\draw[name path=l3, black!65] (0,0) to[out=-65, in=175, looseness=1] node[pos=.97, below, text=black] {$L_A$} (12,-6);

  \fill[red!50!black, name intersections={of=B and l1, name=i, total=\n}]
    \foreach \s in {1,...,\n} {(i-\s) circle (2.5pt) node[above] (\s) {}};


\node[below right=1.3mm, fill=white, inner ysep=0pt, fill opacity=.6, text opacity=1] at (2,0) {$E_k$};
\node[below right=.8mm, fill=white, inner ysep=0pt, fill opacity=.6, text opacity=1] at (1,0) {$E'_k$};
\node[below left=.5mm, fill=white, inner ysep=0pt, fill opacity=.6, text opacity=1] at (i-1) {$p_k$};
\node[above left=1.3mm and 0.5mm, fill=white, inner ysep=0pt, fill opacity=.6, text opacity=1] at (i-2) {$q_k$};
    
  \fill[red!50!black, name intersections={of=B and l2, name=i, total=\n}]
    \foreach \s in {1,...,\n} {(i-\s) circle (0pt) node[above] (\s) {}};
}
 \caption{Curves in the proof of Theorem \ref{t:ssv2}. Given $M$, ellipses $E_k$ are defined so as to cover large portions of $\hat{L}_M^+\cap \{ \im \lambda\leq 0\}$, while at the same time $E_k'$ remains entirely above the curve $L_A$. For large $k$, we have $|p_k| < \varepsilon_3r_k$ and $|q_k| > 2r_k$, so the region $\hat{L}_M^+ \cap \{\im \lambda \leq 0\} \cap \{|\lambda| \leq 2r_k\}$ is covered by the union of the disk $B(0;\varepsilon_3r_k)$ and the convex hull of the ellipse $E_k$.}
    \label{fig:enter-label}
\end{figure}
     
     As $M<T$, $L_M$ intersects $E_k$ at points $p_k$ and $q_k$ corresponding to the points on $E_k$ with angles $\theta_{p_k} \in (\pi, \frac{3\pi}{2})$ and $\theta_{q_k} \in (\frac{3\pi}{2}, 2\pi)$. We claim that for large values of $k$ we will have $|p_k| < \varepsilon_3r_k$. Indeed, since $p_k \in  E_k$, we have
     \begin{align*}p_k  & = (1+\varepsilon_2)r_k + r_k\cosh(\eta_k+i\theta_{p_k})\\
     |p_k|^2& =((1+\varepsilon_2)r_k+r_k\cosh\eta_k\cos\theta_{p_k})^2 + r_k^2\sinh^2\eta_k\sin^2\theta_{p_k}
     \end{align*}
     and as $\varepsilon_2 < \varepsilon_3$ the claim will follow if we can show that $\theta_{p_k}$ is sufficiently close to $\pi$ for large values of $k$. To that end, we equate the two expressions for the imaginary part of $p_k$ as a point of $E_k$ and of $L_M$ to get
     \begin{align*}
     \sin\theta_{p_k}T\log(1+(1+\varepsilon_2)r_k)&  = -M\log(1+|p_k|)\\ &
     = -M\log(1+(1+\varepsilon_2)r_k + r_k\cosh\eta_k\cos\theta_{p_k}) + o(1)
     \end{align*}
     where for the second equality we used that $|p_k| = \re p_k(1+o(1))$ as $k \rightarrow \infty$. From this we obtain
     \[\sin\theta_{p_k} = -\frac{M}{T} + O(1/\log r_k)\]
     as $k \rightarrow \infty$. As $\theta_{p_k} \in (\pi, \frac{3\pi}{2})$, at this point we may assume that we have taken $T$ so large that as $k \rightarrow \infty$, $\theta_{p_k}$ remains close enough to $\pi$ that $|p_k| < \varepsilon_3r_k$ holds, as claimed. An analogous consideration shows that for large values of $k$ we will also have $|q_k| > 2r_k$. Therefore,
     \[\hat{L}_M^+ \cap \{\operatorname{Im}\lambda \leq 0\} \cap \{|\lambda| < 2r_k\} \subset \operatorname{ch}(E_k) \cup B(0;\varepsilon_3r_k),\]
     and hence if $N_k := \#\left(\mathcal{R}_{U+W} \cap \operatorname{ch}(E_k)\right)$, we will have
     \[n_{U+W}(\hat{L}_M^+;2r_k) \leq N_k + \varepsilon_0r_k\]
     for large $k$, where we used (\ref{eq:MZas}) and our definition of $\varepsilon_3$ to get $n_{V+W}(\varepsilon_3r_k) \leq \varepsilon_0r_k$ for large $k$. The assumption (\ref{eq:lba}) then shows that
     \begin{equation}\label{eq:Nk}
         N_k \geq \left (\frac{2}{\pi}|\ch\supp U| + \varepsilon_0\right)r_k
     \end{equation}
     for large $k$.

     Next we need to show that for a large enough $A$, the ellipse
 \[E_k':= \{\lambda = (1+\varepsilon_2)r_k + r_k\cosh(\gamma_k + i\theta): \theta \in [0, 2\pi]\}\]
     will lie entirely above $L_A$ (for large $k$), and we will show that this is the case if $A > T/\varepsilon_1$. Noting that $\sinh x \sim x$ and $\sinh^{-1}x \sim x$ as $x \rightarrow 0+$, we obtain from (\ref{eq:etak}) that
         \[r_k\sinh\gamma_k \sim \frac{T}{\varepsilon_1}\log(1+(1+\varepsilon_2)r_k)\]
     as $k \rightarrow \infty$, so we can find $v > 0$ such that for large $k$,
     \[r_k\sinh\gamma_k < (A-v)\log(1+(1+\varepsilon_2)r_k).\]
     Moreover, for large $k$ and $\lambda \in E'_k$ we have 
     $\re \lambda > r_k(\varepsilon_2+O(\gamma_k^2))$, so that 
     for sufficiently large $k$ and any $\lambda \in E'_k$ we will also have
     \[(A-v)\log(1+(1+\varepsilon_2)r_k) < A\log(1+|\lambda|)\]
     and
     \[|\operatorname{Im}\lambda| \leq r_k\sinh\gamma_k.\]
    Thus  we see that $E_k' \subset \hat{L}_A$ for large $k$ and the claim follows.

     In order to produce a contradiction and conclude the proof, we now apply Lemma \ref{l:koosis} to see that for all $k$
     \begin{multline}\label{eq:lucky}
         \frac{1}{\pi} \int_{-r_k}^{r_k} \frac{\log |D_{U+W}((1+\varepsilon_2)r_k+t)|}{\sqrt{r_k^2-t^2}}dt \\
         \leq \frac{1}{2\pi} \int_0^{2\pi} \log |D_{U+W}((1+\varepsilon_2)r_k + r_k\cosh(\gamma_k+i\theta))|d\theta 
         - N_k(1-\varepsilon_1)\gamma_k.
     \end{multline}
     Using that $|D_{U+W}(\lambda)| \geq C_1$ in $\im \lambda \geq 0, |\lambda| >C_0$ for some constants $C_0, C_1 > 0$ (see Section \ref{sec:preliminaries}), the left hand side in (\ref{eq:lucky}) is uniformly bounded below by $\log C_1$. On the right, we apply Lemma \ref{l:Dub} for the first term and (\ref{eq:Nk}) for the second to find
     \begin{equation*}
     \log C_1 \leq C + \frac{2}{\pi}|\ch\supp U|r_k\sinh\gamma_k-\left(\frac{2}{\pi}|\ch\supp U|+\varepsilon_0\right)(1-\varepsilon_1)r_k\gamma_k
     \leq C-\delta_0r_k\sinh\gamma_k,
     \end{equation*}
     where we used (\ref{eq:d0}) for the latter bound. As the right hand side tends to $-\infty$ as $k \rightarrow \infty$, this contradiction proves the theorem.
\end{proof}

\section{Resonances in sectors away from $\Real$ } \label{s:awayfromR}

The idea behind the proof of Theorem \ref{t:ara} is to use the lower bound given by Lemma \ref{l:Sasymp} and the upper bound given by Lemma \ref{l:Sdiff} to produce the inequality
\[\left|\det S_{V+W}(-\lambda) - \det S_V(-\lambda)\right| < \left| \det S_V(-\lambda)\right|\]
in suitable regions, and then to use Rouché's theorem to relate the number of resonances of the two Schrödinger operators under consideration.

\begin{proof}[Proof of Theorem \ref{t:ara}]
    From Lemma \ref{l:Sdiff} there is an $M < |\ch \supp V|$ such that
    \begin{equation}
        \log |\det S_{V+W}(-\lambda) - \det S_V(-\lambda)| \leq 2M|\operatorname{Im}\lambda|
    \end{equation}
    when $\im \lambda \leq 0$ and $|\lambda|$ is large enough. Choose $\varepsilon > 0$ so that $2M|\sin \alpha| + \varepsilon < 2|\ch \supp V| |\sin \alpha|$ for all $\alpha \in [\theta, \varphi]$. Now let $\ci_V$ be the set from Lemma \ref{l:Sasymp} and choose $R_1 > 0$ so that
    \begin{equation}
        \log|\det S_V(-\lambda)| \geq 2|\ch \supp V| |\operatorname{Im}\lambda| - \varepsilon|\lambda|
    \end{equation}
    for $|\lambda| > R_1$, $\im \lambda \leq 0$, and $\lambda \not\in \ci_V$. Restricting $\lambda$ to $\Bar{\sect}$ in these inequalities, it follows from our choice of $\varepsilon$ that
    \begin{equation}\label{eq:smdb}
        |\det S_{V+W}(-\lambda) - \det S_V(-\lambda)| < |\det S_V(-\lambda)|
    \end{equation}
    when $|\lambda| \geq R_1, \lambda \in \Bar{\sect}, \lambda \not\in \ci_V$.

    Now, as the angles $\theta, \varphi \not\in \Bar{A}_V$, we can find $\delta > 0$ such that $(\theta - \delta, \theta + \delta) \cap \Bar{A}_V = \emptyset$ and $(\varphi - \delta, \varphi + \delta) \cap \Bar{A}_V = \emptyset$ and therefore it is possible to find $R_2 \geq R_1$ such that if $\lambda \in \partial \sect, |\lambda| > R_2$, then $\lambda \not\in \ci_V$.

    Next define 
    $$E_V := \{r > 0: \ci_V \cap \{|\lambda| = r\} \not= \emptyset \}.$$
     This set has finite logarithmic measure (see the remarks following the proof of \cite[Ch. 15, Theorem 1]{Levin96}). If we fix $R \geq \max\{R_1, R_2\}, R \not\in E_V$, then for any $r > R$ with $r \not\in E_V$, the inequality (\ref{eq:smdb}) holds on the boundary of the annular region $\sect \cap \{R < |\lambda| < r\}.$ This shows that
    \[n_{V+W} (\sect;r) - n_{V+W}(\sect; R) = n_V(\sect;r) - n_V(\sect;R)\]
    by an application of Rouché's theorem. Since this holds for any such $r$, the proof is complete. 
\end{proof}

We now give a complementary result which follows immediately and shows that, in addition to their counting functions having the same asymptotics, in sectors away from the real axis resonances of $-\frac{d^2}{dx^2} + V+ W$ do not deviate too much from those of $-\frac{d^2}{dx^2} + V$.

\begin{theorem}\label{t:closeres}
    Let $V \in L_c^\infty(\Real; \mathbb{C})$. Then there exists a union of disks $\ci_V = \bigcup\limits_{i=1}^\infty B(z_i; r_i)$ such that $\mathcal{R}_V \subset \ci_V$ and $\sum\limits_{i=1}^\infty \frac{r_i}{|z_i|} < \infty$, and such that the following holds: For any sector $\sect = \sect(\theta, \varphi)$ with $-\pi < \theta < \varphi < 0$ and any $W \in L_c^\infty(\Real; \mathbb{C})$, with 
    $\supp W \subset (\ch \supp V)^\circ$, there exists $R > 0$ so that
    \[\mathcal{R}_{V+W} \cap \sect \cap \{|\lambda| > R\} \subset \ci_V.\]
\end{theorem}

\begin{proof}
    
    We follow the proof of Theorem \ref{t:ara} to obtain (\ref{eq:smdb}) (which made no assumption that $\theta, \varphi \not\in \Bar{A_V})$. This strict inequality implies that $\det S_{V+W}(-\lambda)$ cannot vanish in $\Bar{\sect} \cap \{|\lambda| \geq R_1\} \cap (\mathbb{C} \setminus \ci_V)$, which implies the conclusion of the theorem.
\end{proof}

\appendix
\section{Applications of Hardy's method}
In this appendix we apply a method of G.H. Hardy which provides a process for asymptotically locating the solutions of certain transcendental equations wherein one first finds 
preliminary expressions for the \emph{possible} asymptotic location of solutions, and subsequently proves that such points actually \emph{are} solutions, by applying Rouché's theorem.

The method was originally used in \cite[Sections 26 and 27]{Hardy04} to locate zeros of a particular entire function, and was subsequently applied repeatedly in \cite{Cartwright30} to study zeros of Fourier transforms of certain classes of functions. More recently, Hardy's method has been applied to study the distribution of resonances for a compactly supported potential in one dimension \cite[Lemma 5]{Zworski87}, \cite[Section 5]{Stepin07}, for radial potentials in higher odd dimensions \cite[Lemma 6]{Zworski89}, \cite[Section 4]{Dinh14}, and for sums of $\delta$-function potentials \cite[Theorem 1]{Datchev22}.

Here we focus on just two situations for which we have applications in Section \ref{s:logc}. The first (Lemma \ref{l:Hm1}) was referenced previously in \cite{Zworski87} but we give an explicit proof here as this version is used repeatedly in this paper. In Lemma \ref{l:Hm2} we give an extension which applies to a more complicated case. In both situations, the transcendental equations to which we apply Hardy's method arise from looking for zeros of analytic functions of a certain form. For $R > 0$ and $0 \leq D_1 < D_2 \leq \infty$, we set 
$$\Omega_R(D_1,D_2) := \{-D_1\log(1+|z|) > \operatorname{Im}z > -D_2\log(1+|z|)\} \cap \{|z| > R\}.$$

\begin{lemma}\label{l:Hm1}
    Let $f$ be a function analytic in $\Omega_R(D_1, D_2)$ and satisfy
    \begin{equation}\label{eq:hma}
        f(z) = 1 - \frac{C}{(iz)^M}e^{iLz}(1+\varepsilon_1(z)) + \varepsilon_2(z) \qquad z \in \Omega_R(D_1, D_2)
    \end{equation}
    where $C \in \mathbb{C}\setminus \{0\}, M>0, \frac{M}{L} \in (D_1, D_2)$ and $\varepsilon_i(z) = o(1)$ uniformly as $|z| \rightarrow \infty$ through $\Omega_R(D_1, D_2)$. Then there exists $R'>R$, $N\in \Natural$ such that the zeros of $f$ in $\Omega_{R'} (D_1, D_2)$ form a sequence
    \begin{equation}\label{eq:hmc}
        z_{\pm n} = \pm\frac{2\pi n}{L} \pm\frac{M\pi}{2L} + \frac{i}{L}\log C - i\frac{M}{L}\log\left(\frac{2\pi n}{L}\right) + \varepsilon_{\pm n},\; n>N
    \end{equation}
    where $\varepsilon_{\pm n} \rightarrow 0$ as $n \rightarrow \infty$.
\end{lemma}

\begin{proof}
    We work in $\Omega_{R_0} (D_1, D_2)$ with $R_0 > R$ large enough so that $|\varepsilon_i(z)| < \frac{1}{2}$ in this region. If $z$ is a 
    zero  of $f$ we find from (\ref{eq:hma}) that
    \begin{equation}\label{eq:at0}
        (1+\varepsilon_1(z))e^{iLz} = \frac{(iz)^M}{C}(1+\varepsilon_2(z)).
    \end{equation}
    Writing $z = x + iy$ and taking magnitudes,
    \begin{equation}\label{eq:magnitude}
        |1+\varepsilon_1|e^{-yL} = \frac{|z|^M}{|C|}|1+\varepsilon_2|
    \end{equation}
    so that
    \[y = -\frac{M}{L}\log|z| + \frac{1}{L}\log|C| - \frac{1}{L}\log\left(\frac{|1+\varepsilon_2|}{|1+\varepsilon_1|}\right).\]
    Now if $D_2 = \infty$, then for any $D$ with $\frac{M}{L} <D$ we can use (\ref{eq:magnitude}) to conclude that there can be only finitely many zeros of $f$ in $\{\operatorname{Im}z < -D\log(1+|z|)\}$, so we may as well assume $D_2$ is finite moving forward. With that agreement, as $|z| \rightarrow \infty$ through $\Omega_{R_0}(D_1, D_2)$ we clearly have $|z| = |x|(1+\varepsilon_3(z))$ where $\varepsilon_3(z) = o(1)$ and therefore
    \begin{equation}\label{eq:yy}
        y = -\frac{M}{L}\log|x| + \frac{1}{L}\log|C| + \varepsilon_4(z)
    \end{equation}
    with $\varepsilon_4(z) = o(1)$.

    Next we divide (\ref{eq:at0}) by (\ref{eq:magnitude}) to get
    \[\left(\frac{1+\varepsilon_1}{|1+\varepsilon_1|}\right)e^{iLx} = \frac{|C|}{C}\frac{(iz)^M}{|z|^M}\left(\frac{1+\varepsilon_2}{|1+\varepsilon_2|}\right)\]
    from which we find that the real part of a zero must satisfy
    \begin{equation}\label{eq:xx}
        iLx_{\pm n} = \pm 2\pi ni - i\arg C + \frac{iM\pi}{2} + iM\arg z_{\pm n} + \varepsilon'_{\pm n}
        = \pm2\pi ni - i\arg C \pm\frac{iM\pi}{2} + \varepsilon_{\pm n}^1
    \end{equation}
    for some $n \in \mathbb{N}$, where $\varepsilon_{\pm n}^1 \rightarrow 0$ as $n \rightarrow \infty$.

    From (\ref{eq:yy}) it is now easy to see that
    \[y_{\pm n} = -\frac{M}{L}\log\left(\frac{2\pi n}{L}\right) + \frac{1}{L}\log|C| + \varepsilon_{\pm n}^2\]
    where $\varepsilon_{\pm n}^2 \rightarrow 0$ as $n \rightarrow \infty$ and together with (\ref{eq:xx}) this shows that any large zero of $f$ must be of the form of one of the points in (\ref{eq:hmc}).

    To conclude the proof we need to show that for large enough $n$ each point of (\ref{eq:hmc}) actually \emph{is} a zero of $f$. To that end, set
    \[\zeta_{\pm n} := \pm\frac{2\pi n}{L} \pm\frac{M\pi}{2L} + \frac{i}{L}\log C - i\frac{M}{L}\log\left(\frac{2\pi n}{L}\right).\]
    If $|\xi| \leq C_0$ for some $C_0 < \frac{\pi}{L}$, then we can obtain
    \[f(\zeta_{\pm n} + \xi) = 1-e^{iL\xi}(1+\delta_{\pm n}^1(\xi)) + \delta_{\pm n}^2(\xi)\]
    for large $n$, where the $\delta_{\pm n}^i$ tend to zero uniformly as $n \rightarrow \infty$ for $|\xi| \leq C_0$. When $|\xi| = C_0$, the function $g(\xi) := 1-e^{iL\xi}$ satisfies $|g(\xi)| \geq C_1$ for some $C_1 > 0$, and therefore it is straightforward to arrange that
    \[|f(\zeta_{\pm n} + \xi)-g(\xi)| < |g(\xi)| \qquad\qquad \text{on} \quad |\xi| = C_0\]
    for large enough $n$. As $g$ has only the simple zero at $\xi = 0$ within $|\xi| \leq C_0$, we see that $f(\zeta_{\pm n}+\bullet)$ has exactly one zero in $|\xi| \leq C_0$ by Rouché's theorem. This must be the point $z_{\pm n}$.
\end{proof}

We now give an extension, determining the zeros of a function $f$ having the more complicated form
\begin{equation}\label{eq:Hm2e}
    f(z) = 1-\frac{C_1}{(iz)^M}e^{iLz}(1+\varepsilon_1(z)) - \frac{C_2}{(iz)^N}e^{iKz}(1+\varepsilon_2(z)) + \varepsilon_3(z)
\end{equation}
when $M$ and $N$ are integers. In Section \ref{s:rsum} we showed how to deal with the case $\frac{M}{L} \not= \frac{N}{K}$ by reducing to an application of the simpler Lemma \ref{l:Hm1}, so we focus here on the case $\frac{M}{L} = \frac{N}{K}$.

\begin{lemma}\label{l:Hm2}
    Let $f$ be a function analytic in $\Omega_R(D_1,D_2)$ and satisfy (\ref{eq:Hm2e}) where $M<N$ are positive integers, $\frac{M}{L} = \frac{N}{K} \in (D_1,D_2)$ and $\varepsilon_i(z) = o(1)$ uniformly as $|z| \rightarrow \infty$ through $\Omega_R(D_1, D_2)$. Then for some $n_0 > 0$ f has N sequences of zeros $\{\{z^{(l)}_{\pm n}\}^\infty_{n=n_0}\}^N_{l=1}$ satisfying
    \[z_{\pm n}^{(l)} = \pm \frac{2\pi Nn}{K} \pm \frac{N\pi}{2K} - i\frac{N}{K}\log a_l - i\frac{N}{K}\log\left(\frac{2\pi Nn}{K}\right) + \varepsilon_{\pm n}^{(l)}\]
    where $\{a_l\}_{l=1}^N$ are the roots of the equation $C_2z^N + C_1z^M - 1 = 0$ repeated to multiplicity and $\varepsilon_{\pm n}^{(l)} \rightarrow 0$ as $n \rightarrow \infty$.
\end{lemma}

\begin{proof}
    Let $a$ be a root of the equation $C_2z^N + C_1z^M - 1 = 0$ and define, for $|\xi| \leq \eta \ll 1$,
    \[\zeta_{\pm n}^a(\xi) := \pm\frac{2\pi Nn}{K} \pm\frac{N\pi}{2K} -i\frac{N}{K}\log(a+\xi) - i\frac{N}{K}\log\left(\frac{2\pi Nn}{K}\right).\]
    From (\ref{eq:Hm2e}) we find that
    \[f(\zeta_{\pm n}^a(\xi)) = 1 - C_1(a+\xi)^M(1+\delta_{\pm n}^1(\xi)) - C_2(a+\xi)^N(1+\delta_{\pm n}^2(\xi)) + \delta_{\pm n}^3(\xi)\]
    where $\delta_{\pm n}^i$ tend to zero uniformly as $n \rightarrow \infty$. If $\eta$ is small enough then when $|\xi| = \eta$, the function $g(\xi) := 1 - C_1(a + \xi)^M - C_2(a+\xi)^N$ satisfies $|g(\xi)| \geq C_\eta$ for some $C_\eta > 0$, and we can arrange,
    \begin{equation}\label{eq:fg}
        |f(\zeta_{\pm n}^a(\xi))-g(\xi)| < |g(\xi)| \qquad\qquad \text{on} \quad |\xi| = \eta
    \end{equation}
    for large enough $n$. It is easy to see that $g$ has a zero at $\xi = 0$ which is of the same order as that of $a$ as a root of $C_2z^N + C_1z^M - 1 = 0$, say $d \geq 1$. Thus, (\ref{eq:fg}) and Rouché's theorem show that $f$ admits $d$ sequences of zeros given by $\{\zeta_{\pm n}^a(\xi_{\pm n}^{(j)})\}_{n=n_0}^\infty$ for $j=1,\dots,d$ and large enough $n_0$, where $|\xi_{\pm n}^{(j)}| < \eta$.

    Now
    \[\zeta_{\pm n}^a(\xi_{\pm n}^{(i)}) = \pm\frac{2\pi Nn}{K} \pm\frac{N\pi}{2K} -i\frac{N}{K}\log a -i\frac{N}{K}\log\left(\frac{2\pi Nn}{K}\right) -i\frac{N}{K}\log\left(1+\frac{\xi_{\pm n}^{(i)}}{a}\right),\]
    so given $\varepsilon > 0$ we can apply the proof with $\eta > 0$ small enough so that for sufficiently large $n$ the last term has modulus less than $\varepsilon$. This shows that the remainders tend to zero with $n$ and gives us the first $d$ sequences in Lemma \ref{l:Hm2}. Repeating the argument for each of the remaining roots concludes the proof.
\end{proof}

The first part of the proof of Lemma \ref{l:Hm1} follows directly from (\ref{eq:hma})--  any large zero must be a point of the sequence (\ref{eq:hmc}).  The 
proof  subsequently shows that these points actually \emph{are} zeros. In contrast, the added complexity in (\ref{eq:Hm2e}) prevents us from concluding that in Lemma \ref{l:Hm2} we have found \emph{all} of the zeros; however, for the application to resonance distribution in Section \ref{s:logc}, using the known asymptotics of the resonance counting function, we are able to conclude that we miss at most a set of resonances with linear density zero.

\bibliographystyle{plain}
\bibliography{sample}{}

\end{document}